\documentclass[preprint,12pt,authoryear]{elsarticle}
\usepackage{etex}
\usepackage[pdftex, colorlinks=true, citecolor=DarkGreen, hyperindex, plainpages=false,
        bookmarksopen, bookmarksnumbered]{hyperref}
\usepackage{amsmath,graphicx,bm,amssymb,amsthm,enumerate,epsfig,psfrag,cases,mathtools}
\renewcommand{\(}{\left(}
\renewcommand{\)}{\right)}
\renewcommand{\[}{\left[}
\renewcommand{\]}{\right]}

\newcommand{\s}{\mathbf{s}}

\newcommand{\f}{\mathbf{f}}
\renewcommand{\S}{\mathbf{S}}

\newcommand{\y}{\mathbf{y}}

\newcommand{\z}{\mathbf{z}}

\newcommand{\R}{\mathbf{R}}
\renewcommand{\r}{\mathbf{r}}

\newcommand{\x}{\mathbf{x}}

\newcommand{\I}{\mathbf{I}}

\newcommand{\C}{\mathbf{C}}
\renewcommand{\P}{\mathbf{P}}
\renewcommand{\t}{\mathbf{t}}
\newcommand{\A}{\mathbf{A}}
\newcommand{\T}{\mathbf{T}}
\newcommand{\U}{\mathbf{U}}
\newcommand{\M}{\mathbf{M}}

\renewcommand{\u}{\mathbf{u}}
\renewcommand{\v}{\mathbf{v}}
\newcommand{\Q}{\mathbf{Q}}

\newcommand{\V}{\mathbf{V}}
\newcommand{\X}{\mathbf{X}}
\newcommand{\Y}{\mathbf{Y}}
\newcommand{\B}{\mathbf{B}}

\newcommand{\e}{\mathbf{e}}

\newcommand{\Tr}[1]{{\rm{Tr}}\left(#1\right)}
\newcommand{\rk}[1]{{\rm{rk}}\left(#1\right)}
\newcommand{\End}[1]{{\rm{End}}}

\renewcommand{\det}[1]{\left|#1\right|}

\newcommand{\sspan}[1]{\langle #1\rangle }

\newtheorem{lemma}{Lemma}
\newtheorem{definition}{Definition}
\newtheorem{theorem}{Theorem}

\newtheorem{corollary}{Corollary}

\newtheorem{rem}{Remark}
\newtheorem{fact}{Fact}

\newcommand{\norm}[1]{\left\lVert#1\right\rVert}

\usepackage{xcolor,fontenc}
\usepackage{inputenc,natbib}
\usepackage{epstopdf,pst-node}
\usepackage{tikz}
\usepackage{enumitem}
\usepackage[justification=centering]{caption}
\usepackage{dsfont}

\biboptions{sort&compress}

\usetikzlibrary{arrows,matrix,positioning}
\usepackage[all]{xy}

\usepackage{accents}
\newcommand*{\dt}[1]{%
  \accentset{\mbox{\large\bfseries .}}{#1}}

\newcommand{\determ}{\operatorname{\mathrm{det}}}
\newcommand{\Gr}{\operatorname{\mathrm G}}
\newcommand{\GL}{\operatorname{\mathrm GL}}
\newcommand{\Hom}{\operatorname{\mathrm Hom}}

\newcommand{\essinf}{\operatorname{\mathrm {ess}\,\mathrm{inf}\;}}

\theoremstyle{remark}

\begin{document}
\begin{frontmatter}
\title{Gaussian and Robust Kronecker Product Covariance Estimation: Existence and Uniqueness}

\author[rvt]{I. Soloveychik\corref{cor1}}
\ead{ilya.soloveychik@mail.huji.ac.il}

\author[rvt]{D. Trushin}

\cortext[cor1]{Corresponding author}
\address[rvt]{Rachel and Selim Benin School of Computer Science and Engineering, \\ The Hebrew University of Jerusalem, Givat Ram, Jerusalem, Israel, 91904}

\begin{abstract}
We study the Gaussian and robust covariance estimation, assuming the true covariance matrix to be a Kronecker product of two lower dimensional square matrices. In both settings we define the estimators as solutions to the constrained maximum likelihood programs. In the robust case, we consider Tyler's estimator defined as the maximum likelihood estimator of a certain distribution on a sphere. We develop tight sufficient conditions for the existence and uniqueness of the estimates and show that in the Gaussian scenario with the unknown mean, $p/q+q/p + 2$ samples are almost surely enough to guarantee the existence and uniqueness, where $p$ and $q$ are the dimensions of the Kronecker product factors. In the robust case with the known mean, the corresponding sufficient number of samples is $\max\[p/q, q/p\] + 1$.
\end{abstract}

\end{frontmatter}

\section{Introduction}
Covariance estimation is a fundamental problem in multivariate statistical analysis. It arises in diverse applications such as signal processing, where knowledge of the covariance matrix is unavoidable in constructing optimal detectors \cite{kelly1986adaptive}, genomics, where it is widely used to measure correlations between gene expression values \cite{schafer2005shrinkage, xie2003covariance, hero2012hub}, and functional MRI \cite{derado2010modeling}. Most of the modern algorithms analyzing social networks are based on Gaussian Graphical Models \cite{lauritzen1996graphical}, where the independences between the graph nodes are completely determined by the sparsity structure of the inverse covariance matrix \cite{banerjee2008model}. In empirical finance, knowledge of the covariance matrix of stock returns is a fundamental question with implications for portfolio selection and for tests of asset pricing models such as the CAPM \cite{bai2011estimating,ledoit2003improved}. Application of structured covariance matrices instead of Bayesian classifiers based on Gaussian mixture densities or kernel densities proved to be very efficient for many pattern recognition tasks, among them speech recognition, machine translation and object recognition in images \cite{dahmen2000structured}. In geometric functional analysis and computational geometry \cite{adamczak2010quantitative} the exact estimation of covariance matrix is necessary to efficiently compute volume of a body in high dimension. The classical problems of clustering and Discriminant Analysis are entirely based on precise knowledge of covariance matrices of the involved populations \cite{friedman1989regularized}, etc.

In many modern applications, data sets are very large with both large number of samples $n$ and large dimension $p$, often with $p \gg n$, leading to the amount of unknown parameters greatly exceeding the number of observations. This high-dimensional regime naturally calls for exploiting or assuming additional structural properties of the data to reduce the number of estimated degrees of freedom. Usually, the specific structures are chosen to be linear or affine. The most popular examples include such models as Toeplitz \cite{snyder1989use, abramovich2007time, wiesel2013time,kavcic2000matrices, asif2005block,fuhrmann1991application, roberts2000hidden,bickel2008regularized,sun2015robust,soloveychik2014tyler}, group symmetric \cite{shah2012group, soloveychik2014groups}, sparse \cite{banerjee2008model, ravikumar2011high, rothman2008sparse}, low rank \cite{fan2008high, johnstone2009consistency, lounici2014high} and many others. Non-linear structures are also quite common in engineering applications. Among them are the Kronecker product model \cite{tsiligkaridis2013convergence, dutilleul1999mle, werner2008estimation, dawid1981some}, linear and sparse inverse covariance structures such as graphical models \cite{friedman2008sparse, yuan2007model} and others.

In this paper we focus on the Kronecker Product (KP) structure, which has recently become an extremely popular model for a variety of applications, such as MIMO wireless communications \cite{werner2007estimation}, geostatistics \cite{cressie2015statistics}, genomics \cite{yin2012model}, multi-task learning \cite{bonilla2007multi}, face recognition \cite{zhang2010learning}, recommendation systems \cite{allen2010transposable}, collaborative filtering \cite{yu2009large} and many others. The KP model assumes a $pq \times pq$ covariance matrix $\bm\Theta_0$ to be a KP of two lower dimensional square matrices, which is denoted by $\bm\Theta_0 = \P \otimes \Q$, where $\P$ and $\Q$ are $p \times p$ and $q \times q$ dimensional positive definite matrices, respectively. Given $\bm\Theta_0$, its factors $\P$ and $\Q$ can only be determined up to a positive scalar. This natural ambiguity is usually treated by fixing scaling of one of the factors as we do below.

Consider the Gaussian setting and assume we are given $n$ independent and identically distributed (i.i.d.) $pq$ dimensional real vector measurements $\x_i \sim \x,\; i=1,\dots,n$, where
\begin{equation}
\x \sim \mathcal{N}(\bm\mu, \bm\Theta).
\end{equation}
Assume the mean $\bm\mu$ is known, while the covariance $\bm\Theta$ is to be estimated. If the number of samples is not less than the ambient dimension, $n \geqslant pq$, the Maximum Likelihood Estimator (MLE) of the covariance parameter almost surely exists and coincides with the Sample Covariance Matrix (SCM),
\begin{equation}
\S = \frac{1}{n}\sum_{i=1}^n (\x_i-\bm\mu)(\x_i-\bm\mu)^\top.
\label{scm_def}
\end{equation}
When the prior knowledge suggests that the true covariance matrix $\bm\Theta_0$ is of the KP structure, it is usually more convenient to cut $\x$ into $q$ columns of height $p$ each to obtain a so-called matrix normal random variable $\X$ \cite{dutilleul1999mle, dawid1981some, gupta1999matrix}. Following \cite{gupta1999matrix}, we denote this law by
\begin{equation}
\X \sim \mathcal{MN}(\M, \P\otimes \Q),
\label{intro_def_g}
\end{equation}
where $\M$ is obtained from $\bm\mu$ by the same reshaping procedure. Assume we are given $n$ i.i.d. matrix samples $\X_i \sim \X,\; i=1,\dots,n$ as in (\ref{intro_def_g}), and want to estimate the covariance matrix factors $\P$ and $\Q$. Here, the MLE solution is no longer given by an explicit formula as in (\ref{scm_def}), moreover, the resulting optimization program is non-convex due to the constraint. Luckily, there exists an alternating optimization approach, which is usually adopted \cite{dutilleul1999mle, lu2005likelihood, lu2004likelihood, werner2008estimation}. This algorithm is often referred to as the Flip-Flop (FF) due to the symmetric updates of the estimates of $\P$ and $\Q$ it produces. Below we show that the obtained constrained program becomes convex under a specific change of metric over the set of positive definite matrices, the so-called geodesic metric \cite{wiesel2012geodesic, wiesel2012unified}, naturally explaining the convergence of the FF and significantly helping to further explore the optimization at hand. We refer to this iterative algorithm as the Gaussian FF (GFF) to distinguish it from another FF scheme introduced later.

In many real world applications the underlying multivariate distribution is actually non-Gaussian and robust covariance estimation methods are required. This occurs whenever the distribution of the measurements is heavy-tailed or a small proportion of the samples exhibits outlier behavior \cite{huber1964robust, maronna1976robust}. Probably the most common extension of the Gaussian family of distributions allowing for treating heavy-tailed populations is the class of elliptically shaped distributions \cite{frahm2004generalized}. Elliptical populations served as the basis for defining a family of the so-called covariance $M$-estimators \cite{maronna1976robust}, of which we focus on Tyler's estimator \cite{tyler1987distribution}. Given $n$ samples $\x_i \in \mathbb{R}^{pq},\; i=1,\dots,n,$ Tyler's covariance matrix estimator is defined as the solution to the fixed point equation
\begin{equation}
\T = \frac{pq}{n}\sum_{i=1}^n \frac{\x_i\x_i^\top}{\x_i^\top\T^{-1}\x_i}.
\label{tylerequ}
\end{equation}
When $\x_i$ are i.i.d. Generalized Elliptically (GE) distributed \cite{frahm2004generalized}, their shape matrix $\bm\Theta$ is positive definite and $n>pq$, Tyler's estimator exists with probability one and is a consistent estimator of $\bm\Theta$ up to a positive scaling factor. The GE family includes as particular cases the generalized Gaussian, the compound Gaussian, the elliptical distributions and many others \cite{frahm2004generalized}. Therefore, Tyler's estimator has been successfully used to replace the SCM in many applications sensitive to outliers or heavy-tailed noise, such as anomaly detection in wireless sensor networks \cite{chen2011robust}, antenna array processing \cite{ollila2003robust}, and radar detection \cite{abramovich2007time, ollila2012complex, bandiera2010knowledge, pascal2008covariance}. 

It was recently demonstrated that Tyler's estimator can be viewed as an MLE of a certain spherical distribution \cite{wiesel2012geodesic, greco2013cramer, soloveychik2014non}. In spite of the fact that the obtained MLE program is not convex, it was later shown to become convex under the geodesic metric change we mentioned above \cite{wiesel2012geodesic, wiesel2012unified}. Both these fundamental discoveries paved a way to the creation of a very useful natural optimization framework characterizing Tyler's estimator, which made possible definition of structured analogs of Tyler's estimator under geodesically convex constraints \cite{soloveychik2014groups, sun2015robust, wiesel2015structured}. The present article extensively uses this new framework to study the existence and uniqueness of the KP constrained Tyler's MLE and the convergence properties of the Robust Flip-Flop (RFF) analog of the GFF algorithm obtained from it. Another very popular in engineering applications example of a linear geodesically convex structure is the so-called group symmetry \cite{shah2012group}. Interestingly, a very recent paper \cite{soloveychik2014groups} utilized the aforementioned optimization methodology to thoroughly investigate the group symmetric Tyler's estimator (STyler) and its performance benefits. It is important to note that multiple geodesically convex constraints can be efficiently superimposed when the underlying physics suggests such prior knowledge, e.g quite often in practice the KP structure is followed by group symmetries, leading to a further decrease in the number of estimated degrees of freedom.

In both Gaussian and robust cases, one of the central questions in high-dimensional environment is: What is the minimal number of samples guaranteeing the existence and uniqueness of the corresponding covariance MLE? As we have already mentioned, in the unconstrained Gaussian MLE it is known that $n=pq$ samples are enough to guarantee the existence and uniqueness almost surely when the mean is known, and $n=pq+1$ when the mean is unknown. This number is, of course, enough in the constrained case as well, however, one would expect that this threshold can be reduced due to the decrease in the estimated number of parameters. Different necessary and sufficient conditions on the number of samples in the Gaussian KP scenario were proposed by a large number of works, see \cite{dutilleul1999mle, srivastava2008models, lu2004likelihood, werner2008estimation, ros2016existence} and references therein. In particular, in \cite{dutilleul1999mle} it was claimed that the number of samples needed to guarantee both the existence and uniqueness of the GFF solution in the unknown mean case, equals $\max\[p/q,q/p\] + 1$. Later, \cite{srivastava2008models} showed that, in fact, $\max[p,q] + 1$ matrix valued measurements are required to guarantee the uniqueness, assuming the estimator exists. In \cite{ros2016existence} the authors showed by a few simple counterexamples that both results from \cite{dutilleul1999mle} and \cite{srivastava2008models} are not correct. Instead, they claimed that ``As yet, there do not seem to be existence results for the case $n \in \[\max\[p/q,q/p\] +1, pq\]$,'' therefore, leaving this question open. 

Unlike the Gaussian setting, in the robust scenario the mean is usually assumed to be known. To the best of our knowledge, in this case the question of the minimal number of samples needed to ensure existence and uniqueness was not properly addressed thus far, except for the trivial necessary condition $n \geqslant \max\[p/q,q/p\]$ stemming from the definition of RFF, as shown below.

The main goal of this paper is to present tight thresholds for the necessary and sufficient conditions guaranteeing existence and uniqueness in both Gaussian and robust cases. Namely, we show that in the Gaussian setting, when the mean is not known, if $n < \max\[p/q,q/p\] +1$, the estimator, even if it exists, is not unique with probability one and, when $n > p/q+q/p + 1$, the solution exists and is unique almost surely. We also provide a discussion explaining that between these bounds the probabilities of existence and non-existence of a unique solution are positive. In the robust case with the mean known, the threshold is $\max\[p/q,q/p\]$. More specifically, if $n$ is less than this number, no unique solution exists, while if $n$ is greater than this value, the estimator almost surely exists and is unique.

The rest of the text is organized as follows. After we introduce notations, we define the Gaussian setting and formulate the problem. Then we discuss the state of the art results concerning the necessary and sufficient conditions for the existence and uniqueness. Section \ref{gauss_m_r_sec} presents our main result, shows the main idea of the proof and demonstrates it on a simple two dimensional example. Then we treat the robust scenario and provide our main contribution there. Finally, Sections \ref{GaussProofSec} and \ref{RobustProofSec} contain all the proofs for the Gaussian and robust cases, correspondingly.

\subsection{Notations}
We deal with real Euclidean spaces denoted by $\mathbb{R}^p$, whose elements are columns written as lower case bolds $\x$. The standard scalar products over such spaces are denoted by $(\cdot,\cdot)$ and the corresponding norms by $\norm{\cdot}$. By $\mathcal{M}_{p \times q}$ we denote the Euclidean space of real $p \times q$ matrices, written as upper case bolds $\X$. $\mathcal{S}(\mathbb{R}^p)$ stands for the linear space of symmetric $p \times p$ matrices and $\mathcal{P}(\mathbb{R}^p) \subset \mathcal{S}(\mathbb{R}^p)$ - for the open cone of positive definite matrices inside it. Note that $\mathcal{S}(\mathbb{R}^p)$ inherits from $\mathcal{M}_{p \times p}$ the natural structure of a Euclidean space with the Frobenius norm. $\I$ denotes the identity matrix of an appropriate dimension. For two spaces $\mathbb{R}^p$ and $\mathbb{R}^q,\;\; \mathbb{R}^p\otimes \mathbb{R}^q$ denotes their tensor product space and for two matrices $\P \in \mathcal{P}(\mathbb{R}^p),\; \Q \in \mathcal{P}(\mathbb{R}^q),\;\; \P\otimes\Q$ denotes their Kronecker product. The spectral norm of a matrix $\P \in \mathcal{P}(\mathbb{R}^p)$ is denoted by $\norm{\P}_2$ and its determinant by $|\P|$. Given a set $X$, its boundary is denoted by $\partial X$. For two sets $X$ and $Y,\; X \times Y$ denotes their Cartesian (direct) product. Given a subset $X$ of a linear space, $\sspan{X}$ denotes its span and $|X|$ - its cardinality. We use standard abbreviation a.s. to denote the almost sure convergence when the measure can be inferred from the context. The symbol $\sim$ replaces saying ``is distributed identically to''.

\section{Gaussian Setting and Problem Formulation}
\label{ProblemSec}
Assume we are given $n$ i.i.d. Gaussian matrix samples
\begin{equation}
\X_1,\dots,\X_n \sim \X, \qquad \X \sim \mathcal{MN}(\M, \P \otimes \Q),
\label{GaussSetting}
\end{equation}
where $\X_i\in \mathcal{M}_{p\times q},\; \P\in \mathcal P(\mathbb R^p)$ and $\Q\in \mathcal P(\mathbb R^q)$. Denote $X = \{\X_1,\dots,\X_n\}$, then, up to an additive constant and scaling, the negative log-likelihood reads as
\begin{equation}
\label{g_ll}
\widetilde f_\mathcal{N}(\M,\P\otimes\Q;X) = \frac{1}{n} \sum_{i=1}^n \Tr{\P^{-1}(\X_i -\M) \Q^{-1}(\X_i -\M)^\top} + \ln|\P\otimes\Q|,
\end{equation}
and is defined over the set $\mathcal M_{p \times q} \times \mathcal M_{\mathcal{N}}$, with
\begin{equation}
\label{MGaussDomain}
\mathcal M_{\mathcal{N}} = \{\P\otimes \Q\mid \P\in \mathcal P(\mathbb R^p),\;\Q\in \mathcal P(\mathbb R^q)\}\subset \mathcal P(\mathcal{M}_{p\times q}).
\end{equation}
Here, the matrix $\P\otimes \Q$ is identified with the positive operator $\P\otimes \Q : \mathcal{M}_{p\times q}\to \mathcal{M}_{p\times q}$ acting by the rule $\X\to \P\X\Q$. The scalar product on $\mathcal{M}_{p\times q}$ is given by $(\A,\B) = \Tr{\A\B^\top}$. Note that $\mathcal M_{\mathcal{N}}$ can be identified with the set
\begin{equation}
\mathcal M_{\mathcal{N}} \cong \{(\P,\Q)\mid \|\P\|_2 = 1\}\subset \mathcal P(\mathbb R^p)\times \mathcal P(\mathbb R^q),
\label{m_n_def}
\end{equation}
where the specific normalization can be chosen arbitrarily.

\begin{rem}
\label{not_rem}
Below we assume the following notational convention: when the set $\mathcal M_{\mathcal{N}}$ is viewed as a subspace of $\mathcal P(\mathcal{M}_{p\times q})$ as in (\ref{MGaussDomain}), the arguments of the negative log-likelihood are written as $\widetilde f_\mathcal{N}(\M,\P\otimes\Q;X)$, however, when $\mathcal M_{\mathcal{N}}$ is identified with a subset of $\mathcal P(\mathbb R^p)\times \mathcal P(\mathbb R^q)$ defined by (\ref{m_n_def}), we write the arguments as $\widetilde f_\mathcal{N}(\M,\P,\Q;X)$, with
\begin{equation}
\widetilde f_\mathcal{N}(\M,\P,\Q;X) = \widetilde f_\mathcal{N}(\M\otimes\P,\Q;X).
\end{equation}
Below we use the same rule for other similar functions and the specific representation of the underlying set can be inferred from the way arguments are written.
\end{rem}

Identification (\ref{m_n_def}) allows us to consider elements $\P\otimes\Q \in \mathcal M_{\mathcal{N}}$ (under proper normalization $\P\otimes\Q = \(\P/\|\P\|_2\)\otimes \(\|\P\|_2\Q\)$ if needed) as pairs $(\P,\Q)$. In addition, it endows $\mathcal M_{\mathcal{N}}$ with a smooth manifold structure making $\widetilde f_\mathcal{N}$ a smooth function over it. The covariance MLE under the KP constraint can now be written as a solution to the following program
\begin{equation}
\min_{\substack{\M \in \mathcal{M}_{p \times q}, \\ (\P,\Q)\in \mathcal M_{\mathcal{N}}}} \widetilde f_\mathcal N(\M,\P,\Q;X).
\end{equation}
As in the unconstrained Gaussian case, this program decouples into minimization w.r.t. (with respect to) the unknown mean $\M$, yielding
\begin{equation}
\widehat\M = \frac{1}{n}\sum_{i=1}^n\X_i,
\label{avg_def}
\end{equation}
and minimization w.r.t. to $\P$ and $\Q$. Note that $\ln|\P\otimes\Q| = q\ln|\P|+p\ln|\Q|$, and denote $\Y_i = \X_i - \widehat\M$, then the first-order optimal conditions for $\P$ and $\Q$ read as
\begin{equation}
\begin{cases}
\P = \displaystyle \frac{1}{qn}\sum_{i=1}^n \Y_i \Q^{-1}\Y_i^\top, \\
\Q = \displaystyle \frac{1}{pn}\sum_{i=1}^n \Y_i^\top \P^{-1}\Y_i.
\end{cases}
\label{gff_sys}
\end{equation}
There does not exist a closed form analytic solution to (\ref{gff_sys}), therefore, it is usually solved numerically via the so-called Flip-Flop (FF) iterative scheme \cite{dutilleul1999mle}, which we call the Gaussian FF (GFF). The GFF algorithm works as follows. Starting from an initial guess $(\P_0,\Q_0) \in \mathcal M_{\mathcal N}$ for $(\P,\Q)$, we plug it into the right-hand side of (\ref{gff_sys}) and get a new pair $(\P_1,\Q_1)$. After we normalize this pair to make it formally belong to $\mathcal M_{\mathcal N}$, we apply the procedure to $(\P_1,\Q_1)$ instead $(\P_0,\Q_0)$ and so on as shown in diagram (\ref{ff-diag}),
\begin{equation}
\begin{cases}
\widetilde \P_{j+1} = \displaystyle \frac{1}{qn}\sum_{i=1}^n \Y_i \Q_j^{-1}\Y_i^\top, \\
\widetilde \Q_{j+1} = \displaystyle \frac{1}{pn}\sum_{i=1}^n \Y_i^\top \P_j^{-1}\Y_i, \\
\P_{j+1} = \displaystyle \frac{1}{\|\widetilde \P_{j+1}\|_2} \widetilde \P_{j+1}, \\ 
\Q_{j+1} = \displaystyle \|\widetilde \P_{j+1}\|_2\widetilde \Q_{j+1}.
\end{cases}
\label{gff_norm}
\end{equation}
The consecutive pairs $(\P_j,\Q_j)$ serve as approximations to the true solution, therefore, the convergence of this sequence to the minimum of $\widetilde f_\mathcal N(\widehat{\M},\P,\Q;X)$ on $\mathcal M_\mathcal N$ (if it exists) constitutes one of the central topics of our paper. We start form listing the existing results on the questions at hand.

\section{Existence, Uniqueness and Convergence: State of the Art}
\label{StateOfArtSec}
Having derived the G-CARMEL (Gaussian KRonecker product MLE) solution and obtained an iterative scheme for its calculation, our next goal is to determine the necessary and sufficient conditions for its existence and uniqueness and for the convergence of the GFF procedure. The only parameter under our control is the required number of i.i.d. normal samples, therefore, below we focus on the question: How many measurements one needs to guarantee existence, uniqueness and convergence almost surely?

\begin{itemize}
[leftmargin=*]
\item {\bf{Existence}}. We start from the sufficient conditions. It was claimed in \cite{dutilleul1999mle} that $\max[p/q,q/p]+1$ samples are needed for the existence and uniqueness of the MLE solution in the Gaussian case. However, it was later shown by a counterexample \cite{ros2016existence} that the uniqueness does not follow from this condition. In addition, the authors of \cite{ros2016existence} write that ``Moreover, it is not known whether it [this condition] guaranties existence, because it is not sufficient to show that all updates of the FF algorithm have full rank as is done in \cite{dutilleul1999mle}. It could still happen that the sequence of updates converges (after infinitely many steps) to a Kronecker product that does not have a full rank with the likelihood converging to its supremum.'' It is also claimed in \cite{ros2016existence} that no less than $pq + 1$ samples are required to ensure the existence a.s. This number of measurements coincides with the one needed in the unconstrained case and does not explore the KP structure. Finally, the authors of \cite{ros2016existence} conclude that nothing can be said regarding the existence, if the number of samples lies inside the interval $n \in \[\max[p/q,q/p +1, pq\]$.

The necessary conditions were also treated in \cite{dutilleul1999mle}, where the author claims that if the estimator exists, then $n \geqslant \max\[p/q,q/p\] + 1$. This is clearly true, since if the number of samples is less than this threshold, at least one of the right-hand sides in (\ref{gff_sys}) is rank deficient and cannot be invertible.

\item {\bf{Uniqueness}}.
As summarized in \cite{ros2016existence}, the author of \cite{dutilleul1999mle} claims that the G-CARMEL is unique whenever $n \geqslant \max\[p/q,q/p\]+1$. Later, the authors of \cite{srivastava2008models} stated that indeed $n \geqslant \max[p,q] + 1$ is needed to ensure the uniqueness. Here again, \cite{ros2016existence} succeeded to find counterexamples showing that both these bounds do not guarantee uniqueness. Moreover, the paper \cite{ros2016existence} describes the exact parts of the proofs which seem to contain mistakes, however, the correct lower bounds on the number of required samples are not provided. In fact, to the best of our knowledge, tight sufficient conditions for the uniqueness have not been reported so far.

\item {\bf{Convergence of the Gaussian Flip-Flop Algorithm}}.
The last question regarding the G-CARMEL on which we focus, is the convergence of the GFF iterative scheme. In \cite{dutilleul1999mle} the author establishes the convergence of the GFF technique empirically. He claims that if the limiting points of the sequences $\{\widehat{\Q}_j\}$ and $\{\widehat{\P}_j\}$ (in his notations) do not depend on the initial point, and an additional condition on the second derivatives of the objective is satisfied at the limiting points, then these limits provide the G-CARMEL solution. If such limiting points are not uniquely determined, but rather depend on the initial guesses, they must provide local extrema of the likelihood function. Unfortunately, this empirical approach can hardly be applied in practice and does not provide a strict criterion for the convergence of the GFF.

The authors of \cite{lu2005likelihood, lu2004likelihood} claim that when the number of samples is $n \geqslant pq + 1$, the GFF is guarantied to converge, however they doubt if it really converges to the MLE, since the ``parameter space of $(p,q)$-separable covariance matrices is not convex''. They emphasize that for some values of $n$ the algorithm can converge to many different estimates, depending on the starting value. Finally, they conjecture that for $n$ large enough ``the limit point of the GFF can safely be regarded as the unique MLE'' without proving this statement.

In \cite{werner2008estimation} theoretical asymptotic properties of the GFF algorithm are considered and the algorithm's performance for small sample sizes is investigated with a simulation study.
\end{itemize}

The main contributions of the Gaussian part of our paper consist in 
\begin{itemize}
\item proving tight sufficient and necessary conditions for the a.s. existence and uniqueness of the G-CARMEL estimate,
\item showing that the sufficient conditions also imply convergence of the GFF iterations to the unique solution starting from any initial guess.
\end{itemize}

\section{Main Results and Arguments}
\label{gauss_m_r_sec}
In this section we state our main result in the normal case, give the intuition behind the proof argument and demonstrate our technique on a simple example in a low dimension.

\subsection{The Main Statement}
\begin{theorem}
\label{main_g}
Assume $X = \{\X_1,\dots\X_n\}$ are independently sampled from a continuous distribution over $\mathcal{M}_{p \times q}$ and consider the problem of minimizing $\widetilde f_{\mathcal N}(\M,\P,\Q;X)$ over $\mathcal{M}_{p \times q} \times \mathcal M_{\mathcal{N}}$, then
\begin{enumerate}[leftmargin=0.75cm]
\item if $n < \max\[p/q,q/p\] + 1$, there is no unique minimum,
\item if $n> p/q+q/p + 1$, there is a unique minimum a.s.,
\item if $n> p/q+q/p + 1$, the GFF converges starting from any point of $\mathcal M_{\mathcal{N}}$ to this unique minimum a.s.
\end{enumerate}
\end{theorem}
\begin{proof}
This is a direct corollary of Theorem~\ref{theorem:continuityExpVar} from Section \ref{main_th_pr_sec}.
\end{proof}

\begin{rem}
The statement of the theorem is valid for any continuous distribution and is not limited to the Gaussian ones. Indeed, the claim does not assume any specific statistical model and does not provide statistical guaranties (e.g. consistency), but rather treats the questions of the existence and uniqueness of the minimum.
\end{rem}

\begin{rem}
Note the gap between items 1) and 2) containing one (when $p \neq q$) or two (when $p=q$) integer points which cannot be eliminated. We discuss this phenomenon below in more detail.
\end{rem}

\subsection{Sketch of the The Proof}
In this section we discuss the main building blocks of the proof of Theorem~\ref{main_g} leaving the technical details to Section \ref{GaussProofSec}.
\begin{itemize}
[leftmargin=*]
\item {\textbf{Reduction to the Centered Case}}.
Let
\begin{equation}
\widetilde g_{\mathcal N}(\P\otimes\Q;X) = \widetilde f_{\mathcal N}(0,\P\otimes\Q;X),
\end{equation}
then minimization of $g_\mathcal N$ over $\mathcal M_{\mathcal{N}}$ does not require optimization w.r.t. the mean parameter. The general case with the unknown expectation can be reduced to it through the following observation. Given a family $X = \{\X_1,\dots,\X_n\} \subset \mathcal{M}_{p \times q}$ of $n$ random $p\times q$ matrices, there always exists another family $Y=\{\Y_1,\ldots,\Y_{n-1}\} \subset \mathcal{M}_{p \times q}$ of $n-1$ random matrices, such that
\begin{equation}
\widetilde g_{\mathcal N}(\P\otimes\Q;Y) = \widetilde f_{\mathcal N}(\widehat{ \M},\P\otimes\Q;X).
\end{equation}
Lemma~\ref{lemma:continuityExpVar} from Section \ref{GaussProofSec} shows why this is true and justifies our transition to the zero mean case. In the remainder of this section we treat the zero mean setting.

\item {\textbf{Necessary Conditions}}.
Since we require the solution to be composed of invertible matrices $\P$ and $\Q$, (\ref{gff_sys}) must hold at the extremum point. Note that its right-hand side is not invertible for $n<\max[p/q,q/p]$, therefore, returning to the non-centered case and compensating for this by adding one sample, yields item 1) of Theorem~\ref{main_g}.

\item {\textbf{Sufficient Conditions}}.
To derive the sufficient conditions, in Section \ref{GaussProofSec} we change the parametrization of $\widetilde g_\mathcal N$ by 
\begin{equation}
g_\mathcal N(\P\otimes\Q;X) = \widetilde g_\mathcal N (\P^{-1}\otimes\Q^{-1};X) = \frac{1}{n}\sum_{i=1}^n(\P\X_i\Q,\X_i) - \ln\det{ \P\otimes \Q},
\label{g_n_ex}
\end{equation}
and introduce a specific metric over $\mathcal P(\mathbb R^p \times \mathbb R^q)$, w.r.t. which the set $\mathcal M_{\mathcal{N}}$ and the function $g_\mathcal N(\P\otimes\Q;X)$ are convex. The desired solution exists and is unique if and only if $g_\mathcal N$ continuously tends to $+\infty$ on the boundary as shown in Lemma~\ref{lemma:hatgProp}, in which case it is also strictly convex. Theorem~\ref{theorem:continuity} then demonstrates that this happens a.s. w.r.t. the distribution of $X$ when $n > p/q + q/p$. In the next section we demonstrate the reasoning behind these claims by exploiting the $p=q=2$ case in more detail.

\item {\textbf{Convergence of the GFF}}. Suppose we are given a pair of matrices $(\P_0,\Q_0) \in \mathcal{M}_\mathcal N$ and use (\ref{gff_norm}) to generate the sequence
\begin{equation}
\xymatrix{
	{\P_0}\ar[dr]&{\P_1}\ar[dr]&{\P_2}\ar[dr]&{\ldots}\\
	{\Q_0}\ar[ur]&{\Q_1}\ar[ur]&{\Q_2}\ar[ur]&{\ldots}
}
\label{ff-diag}
\end{equation}
Here, the successive iterates $\P_j,\Q_{j+1},\P_{j+2},\ldots$ are obtained by minimizing $g_\mathcal N$ w.r.t. $\Q$ when $\P_j$ is fixed and similarly by minimizing w.r.t. $\P$ when $\Q_j$ is fixed, etc. As we have mentioned in the previous paragraph, $g_\mathcal N$ is a.s. strictly convex and tends to $+\infty$ on the boundary when $n > p/q + q/p$. This guarantees a decrease of the target function on each iteration and the convergence of the sequence $(\P_j,\Q_{j+1})$ to the unique minimum, hence, $(\P_j,\Q_j)$ converges as well.
\end{itemize}
Let us now illustrate the main arguments by a simple low dimensional example.

\subsection{$p=q=2$ Case Study}
\label{ex_case_st}
Assume $\mathbb R^p = \mathbb R^q = \mathbb R^2, \; X = \{\X_1,\dots,\X_n\}\subset \mathcal M_{2\times 2}$ consists of $n$ matrices and we deal with the case of zero mean. We are going to show a bit more than we have announced in the previous section, namely, we will prove that
\begin{enumerate}
\item If $n = 1$, the set of minima is non empty and forms a submanifold of dimension $3$ with probability one. In particular, a minimum exists but is not unique.

\item If $n = 2$, there exists a polynomial $D(\X_1,\X_2)$ such that
\begin{itemize}
\item if $D(\X_1,\X_2)\geqslant 0$, there is no unique minimum of $g_\mathcal N$ over $\mathcal M_\mathcal N$,

\item if $D(\X_1,\X_2)<0$, there is a unique minimum of $g_\mathcal N$ over $\mathcal M_\mathcal N$,
\end{itemize}
both happening with positive probabilities.

\item If $n > 2$, there is a unique minimum of $g_\mathcal N$ over $\mathcal M_\mathcal N$.
\end{enumerate}

As explained in Section \ref{metricSubSec}, the set $\mathcal M_{\mathcal{N}} \subset \mathcal P(\mathbb R^2 \times \mathbb R^2)$ is convex w.r.t. to a specific metric change. In addition, Lemma~\ref{lemma:hatgProp} demonstrates that the solution to the optimization at hand exists and is unique if and only if $g_\mathcal N$ continuously tends to $+\infty$ on the boundary, which we use below.

\textbf{1) $n=1$ case.} Here, both equations in (\ref{gff_sys}) (after replacing $\P$ and $\Q$ by their inverses) become identical to
\begin{equation}
\P^{-1} = \frac{1}{2}\X_1\Q\X_1^\top.
\end{equation}
This equation defines a submanifold $\mathcal M_m \subset \mathcal M_\mathcal N$ isomorphic to $\mathcal P(\mathbb R^2)$ containing $\Q$-s, and, therefore, having
dimension $3$. A straightforward computation shows that the value of $g_\mathcal N$ is constant on $\mathcal M_m$. Since $g_\mathcal N$ is convex, all points of $\mathcal M_m$ are minima.

\textbf{2) $n=2$ case.} It turns out that the critical question defining the behavior of the solution here is whether there exists a vector $\t\in \mathbb R^2$ such that $\X_1\t$ and $\X_2\t$ are parallel. If the answer is negative, the minimum exists and is unique, otherwise, if it exists, it is not unique. As Lemma~\ref{lemma:hatgProp} item 4) shows below, such vector $\t$ does not exist if and only if $g_\mathcal N$ tends to $+\infty$ on the boundary of $\mathcal M_{\mathcal{N}}$. Next we explain the reasoning in more detail and explicitly construct such $\t$.

Consider a sequence $\mathcal M_{\mathcal{N}} \ni \{\(\P_j,\Q_j\)\}\to\partial\mathcal M_{\mathcal{N}}$, meaning that either $\P_j$ tend to a singular matrix or the norms of $\Q_j$ are unbounded (or both). Below we suppress the $j$ indexing of the sequences to simply notations. In other words we distinguish between two cases: a) either $\|\P\|_2$ and $\|\Q\|_2$ are bounded or b) we may assume that in some appropriately chosen bases $\{\s_1,\s_2\}$ and $\{\t_1,\t_2\}$ for $\P$ and $\Q$, respectively, we have
\begin{equation}
\P =
	\begin{pmatrix}
		{\alpha}&{0}\\
		{0}&{1}
	\end{pmatrix}
\quad
\Q =
	\begin{pmatrix}
		{\mu}&{0}\\
		{0}&{\eta}
	\end{pmatrix},
\label{div_seq}
\end{equation}
where $\alpha \leqslant 1$ and we assume (after swapping the eigenbasis of $\Q$, if necessary) that $\mu \to +\infty$ not slower than $\eta$ (if $\eta$ is bounded this is vacuously true). 

In the first case, when the spectral norms are bounded, the trace term of $g_\mathcal N$ is bounded. Since in this scenario at least one of the matrices must tend to a singular one, $\ln\det{\P\otimes \Q}\to -\infty$, implying $g_\mathcal N\to +\infty$. In the second case we have sequence (\ref{div_seq}) and note that the logarithmic term of $g_\mathcal N$ has a summand tending to $-\infty$ with the rate not greater than $\ln\mu$. Assume $\X_1\t_1$ and $\X_2\t_1$ are not parallel for all $\t_1$, then at least one of $\X_i$-s has a non zero $(2,1)$ element. Suppose this is $\X_1= \(\begin{smallmatrix}x_{11}&x_{12}\\ x_{21}&x_{22}\end{smallmatrix}\)$ with $x_{21}\neq 0$, then the scalar products part of $g_\mathcal N$ is not less than $\frac{1}{2}|x_{21}|^2\mu$. Hence, this part of $g_\mathcal N$ tends to $+\infty$ faster than the negative part and totally $g_\mathcal N\to +\infty$ on the sequence at hand.

Now suppose that there does exist a vector $\t\in \mathbb R^2$ such that $\X_1\t$ and $\X_2\t$ are collinear. Normalize $\t$ and form an orthonormal basis $\{\t,\t'\}$ in the space of $\Q$. After this, normalize $\X_1\t$, which we denote by $\s$, and complete it to an orthonormal basis $\{\s,\s'\}$ in the space of $\P$. In these bases each $\X_i$ reads as (here we omit index $i$ in matrix elements for simplicity)
\begin{equation}
	\X_i =
	\begin{pmatrix}
		{x_{11}}&{x_{12}}\\
		{0}&{x_{22}}
	\end{pmatrix}.
\end{equation}
Now define a new sequence in the chosen bases
\begin{equation}
	\P=
	\begin{pmatrix}
		{\frac{1}{\mu}}&{0}\\
		{0}&{1}
	\end{pmatrix},
	\quad
	\Q =
	\begin{pmatrix}
	{\mu}&{0}\\
	{0}&{1}
	\end{pmatrix},
\end{equation}
with $\mu\to +\infty$. Then $\P\otimes \Q$ tends to the boundary of $\mathcal M_{\mathcal{N}}$,
\begin{equation}
\ln\det{\P\otimes \Q} = 2\ln\det{\P} + 2\ln\det{\Q}
=2\ln\det{\P\Q} = 0,
\end{equation}
and for each $\X_i$,
\begin{equation}
(\P\X_i\Q, \X_i) = x_{12}^2 + x_{21}^2 + x_{22}^2 \frac{1}{\mu}.
\end{equation}
Hence, $g_\mathcal N$ is bounded on the sequence $\{\(\P,\Q\)\}$ and we are done with this case.

Next we derive a condition on $\X_1$ and $\X_2$ telling whether such a mutual $\t$ exists, which will suggest us the probability of such event. Let in the original bases $\X_1,\; \X_2$ and $\t$ read as
\begin{equation}
	\X_1 =
	\begin{pmatrix}
		{x}&{y}\\
		{u}&{v}\\
	\end{pmatrix},
	\quad
	\X_2 =
	\begin{pmatrix}
		{a}&{b}\\
		{c}&{d}\\
	\end{pmatrix},
	\quad
	\t =
	\begin{pmatrix}
		{\alpha}\\
		{\beta}\\
	\end{pmatrix}.
\end{equation}
We look for all triples $(\X_1,\X_2,\t)$ such that $\X_1\t$ and $\X_2\t$ are collinear, which is equivalent to
\begin{equation}
	\begin{vmatrix}
		{\alpha x + \beta y}&{\alpha a + \beta b}\\
		{\alpha u + \beta v}&{\alpha c + \beta d}\\
	\end{vmatrix}
	=0.
\end{equation}
The latter can be written as
\begin{equation}
	\begin{vmatrix}
		{x}&{a}\\
		{u}&{c}\\
	\end{vmatrix}
	\alpha^2 + \(
	\begin{vmatrix}
		{x}&{b}\\
		{u}&{d}\\
	\end{vmatrix}
	+
	\begin{vmatrix}
		{y}&{a}\\
		{v}&{c}\\
	\end{vmatrix}
	\) \alpha \beta +
	\begin{vmatrix}
		{y}&{b}\\
		{v}&{d}\\
	\end{vmatrix}
	\beta^2
	=0.
\end{equation}
Calculate the discriminant of this quadric
\begin{equation}
	D(\X_1,\X_2) = \(
		\begin{vmatrix}
			{x}&{b}\\
			{u}&{d}\\
		\end{vmatrix}
		+
		\begin{vmatrix}
			{y}&{a}\\
			{v}&{c}\\
		\end{vmatrix}
	\)^2
	-
	4
	\begin{vmatrix}
		{x}&{a}\\
		{u}&{c}\\
	\end{vmatrix}
	\begin{vmatrix}
		{y}&{b}\\
		{v}&{d}\\
	\end{vmatrix}.
\end{equation}
Note that $D(\X_1,\X_2)$ is a non-zero polynomial and there is $\t\in \mathbb{R}^2$ with the required properties if and only if $D\geqslant 0$. If the density of the distribution of $(\X_1,\X_2)$ is a.s. non-zero, then clearly $D\geqslant 0$ and $D<0$ both hold with non-zero probabilities.

\textbf{3) $n>2$ case.} Here a similar computation shows that we a.s. cannot find a vector $\t$ such that $\X_i\t$ are collinear, therefore the above arguments imply the existence and uniqueness of the minimum.

\subsection{Remarks}
To summarize, the answer to the existence and uniqueness question can be completely described in terms of the following indicator variable:
\begin{equation}
\zeta(X) = \mathop{\essinf}_{\u\in \mathbb{R}^2 \setminus 0}\dim\(\sum_{\X_i\in X} \X_i \langle \u\rangle\).
\end{equation}
When $n=1$, $\zeta(X) = 1$ a.s., in the case $n>2$, $\zeta(X) = 2$ a.s., and these two situations correspond to the uniqueness and non uniqueness. When $n=2$, we have
\begin{equation}
	\zeta(X) =
	\begin{cases}
		1, &\text{if $D \geqslant 0$},\\
		2, &\text{if $D <0$},
	\end{cases}
\label{ind_f}
\end{equation}
where both events happen with non-zero probabilities, i.e. $\zeta(X)$ is not a constant a.s. This intuitively explains the one sample gap between the necessary and sufficient conditions in Theorem~\ref{main_g}.

In arbitrary dimension the ideas described above generalize as following. In order to guarantee the desired asymptotic behavior of the target function, our aim would be to avoid the following situation: there is a random subspace $U\subseteq \mathbb R^q$ such that the dimension of $\sum_{i=1}^n \X_i U$ is less than $\min [n\dim U, p]$ with non-zero probability. As the proof of Theorem~\ref{main_g} shows, when the number of samples satisfies the required condition, such event will a.s. not happen.

Let us focus on the gap between the necessary and sufficient conditions appearing in the statement of Theorem, \ref{main_g} 
\begin{equation}
\mathcal{I} = \[\max\(\frac{p}{q},\frac{q}{p}\) +1,\;\;\frac{p}{q}+\frac{q}{p} +1\].
\end{equation}
This interval contains $2$ points in case $p=q$ and only $1$ point otherwise. The same argument as in the example above (note that in the example we considered the known mean case, therefore the $n=1$ case considered there would correspond to $n=2$ here) shows that when $p=q$ and $n=\max\[p/q,q/p\]+1 = 2 \in \mathcal{I}$, there are multiple minima. Therefore, there is only one untreated integer point $p/q+q/p +1$ left inside the interval $\mathcal{I}$. However, we do not investigate deeply the behavior of this remaining value due to the following reason. As the two dimensional example above suggests (this corresponds to the case $n=2$ case in the example), in this case uniqueness and non-uniqueness happen with non-zero probabilities, making the analysis hard. Since there is only one untreated point left and  the treatment involves quite non-trivial calculations, the game does not worth the candle. We believe that in general dimensions this missing point of the interval exhibits the same behavior, and both events ``existence and uniqueness'' and ``existence and non-uniqueness'' happen with non-zero probabilities.

\section{Robust Kronecker Product Covariance Estimation}
\subsection{Tyler's Estimator}
As we have already explained in the Introduction, when robust covariance estimation is considered, the most popular tool used by practitioners is the so-called family of covariance $M$-estimators introduced by R.A. Maronna \cite{maronna1976robust}. We focus on a distribution-free member of this class introduced by D.E. Tyler and named after him \cite{tyler1987distribution}. Tyler's covariance estimator, given by formula (\ref{tylerequ}), can be equivalently defined as a covariance parameter MLE of a certain spherical distribution \cite{greco2013cramer, soloveychik2014non} as follows.
\begin{definition}
\label{def}
Assume $\bm\Theta_0 \in \mathcal{P}(\mathbb{R}^p)$, then
\begin{equation}
p(\x) = \frac{\Gamma(p/2)}{2\sqrt{\pi}^p} \frac{1}{\sqrt{|\bm\Theta_0|}(\x^H  \bm\Theta_0^{-1}\x)^{p/2}}
\label{tyler_distr}
\end{equation}
is a probability density function of a vector $\x \in \mathbb{R}^p$ lying on a unit sphere. This distribution is usually referred to as the Real Angular Central Elliptical (RACE) distribution \cite{greco2013cramer}, and we denote it as $\x \sim \mathcal{U}(\bm\Theta_0)$. The matrix $\bm\Theta_0$ is referred to as a shape matrix of the distribution.
\end{definition}
RACE distribution is closely related to the class of Generalized Elliptical (GE) populations, which includes Gaussian, compound Gaussian, elliptical, skew-elliptical, RACE and other distributions \cite{frahm2007tyler}. An important property of the GE family is that the shape matrix of a population does not change when the vector is divided by its Euclidean norm \cite{frahm2004generalized, frahm2007tyler}. After normalization any GE vector becomes RACE distributed. This allows us to treat all these distributions together using Tyler's estimator, which is the MLE of the shape matrix parameter in RACE populations and is unbiased when a specific scaling is fixed \cite{greco2013cramer, soloveychik2014non}.

\subsection{Robust Setting and Problem Formulation}
In order to proceed to the KP structured robust covariance estimation, we introduce the following setting. Assume we are given $n$ i.i.d. centered real $p\times q$ matrix measurements $X=\{\X_1,\dots,\X_n\}$ and our goal is to determine what is the minimal number of samples $n$ needed to ensure the existence and uniqueness of Tyler's estimator under the KP constraint. We use the MLE formulation of Tyler's estimator and consider the corresponding optimization program. Specifically, we search for positive definite $\P$ and $\Q$ minimizing the target
\begin{equation}
\widetilde f_\mathcal{E}(\P\otimes\Q;X) = \frac{1}{pq}\ln|\P \otimes \Q|
+ \frac{1}{n} \sum_{i=1}^n \ln\(\Tr{\P^{-1}\X_i\Q^{-1}\X_i^\top}\),
\label{e_ll}
\end{equation}
which is a robust version of the G-CARMEL estimator and is named R-CARMEL.

The target $\widetilde f_\mathcal{E}$ is naturally defined over $\mathcal M_{\mathcal{N}}$ introduced in (\ref{MGaussDomain}) and Remark \ref{not_rem} applies here as well, therefore, we use the same notational convention. In addition, $\widetilde f_\mathcal{E}$ is scale invariant $\widetilde f_\mathcal{E}(\lambda \P\otimes\Q;X) = \widetilde f_\mathcal{E}(\P\otimes\Q;X)$, hence, we rather consider $\widetilde f_\mathcal{E}$ over
\begin{equation}
\mathcal M_{\mathcal E} = \mathcal M_{\mathcal{N}}/\{
\P\otimes \Q \sim \lambda \P\otimes \Q,\; \lambda > 0\}.
\end{equation}
The induced map $\mathcal M_{\mathcal{N}} \to \mathcal M_{\mathcal E}$ is surjective and has no critical points. The composition $\mathcal P(V)\times \mathcal P(U)\to \mathcal M_{\mathcal{N}} \to \mathcal M_{\mathcal{E}}$ admits a section, thus, we may treat $\mathcal M_{\mathcal{E}}$ as
\begin{equation}
\mathcal M_{\mathcal{E}} \cong  \{(\P,\Q)\mid \|\P\|_2 =\|\Q\|_2 = 1\}\subset \mathcal P(\mathbb R^p)\times \mathcal P(\mathbb R^q),
\label{m_e_def}
\end{equation}
which provides $\mathcal M_{\mathcal{E}}$ with a smooth manifold structure. The reason we still use $\mathcal M_{\mathcal{N}}$ is the metric it possesses, whereas we cannot provide $\mathcal M_{\mathcal{E}}$ with a similar metric. Below we demonstrate that the same changes of parametrization and metric as we utilized in the Gaussian case, make $\widetilde f_\mathcal{E}$ convex and significantly simplify the treatment. On the other hand, we need $\mathcal M_{\mathcal{E}}$ when we talk about the uniqueness of the extremum, since there is no uniqueness of the minimum of $\widetilde f_\mathcal{E}$ over $\mathcal M_{\mathcal{N}}$ due to the scaling ambiguity.

Minimization of $\widetilde f_{\mathcal{E}}$ w.r.t. $\P$ and $\Q$ yields a critical point defined by the following system
\begin{equation}
\begin{cases}
\P = \displaystyle \frac{1}{qn}\sum_{i=1}^n \frac{\X_i \Q^{-1}\X_i^\top}{\Tr{\P^{-1}\X_i \Q^{-1}\X_i^\top}}, \\
\Q = \displaystyle\frac{1}{pn}\sum_{i=1}^n \frac{\X_i^\top \P^{-1}\X_i}{\Tr{\P^{-1}\X_i \Q^{-1}\X_i^\top}}.
\end{cases}
\label{eff_sys}
\end{equation}
Similarly to the Gaussian case, there does not exist a closed form solution to this system, and an iterative solution is required which we call the Robust Flip-Flop (RFF). It is also a descent algorithm and converges starting from any initial point due to a similar reasoning. If one wants to remain inside the set $\mathcal M_{\mathcal{E}}$ on each iteration, he has to normalize the iterates on each step
\begin{equation}
\begin{cases}
\widetilde\P_{j+1} = \displaystyle \frac{1}{qn}\sum_{i=1}^n \frac{\X_i \Q_j^{-1}\X_i^\top}{\Tr{\P_j^{-1}\X_i \Q_j^{-1}\X_i^\top}}, \\
\widetilde\Q_{j+1} = \displaystyle\frac{1}{pn}\sum_{i=1}^n \frac{\X_i^\top \P_j^{-1}\X_i}{\Tr{\P_j^{-1}\X_i \Q_j^{-1}\X_i^\top}}, \\
\P_{j+1} = \displaystyle \frac{\widetilde\P_{j+1}}{\|\widetilde\P_{j+1}\|_2}, \\ 
\Q_{j+1} = \displaystyle \frac{\widetilde\Q_{j+1}}{\|\widetilde\Q_{j+1}\|_2}.
\end{cases}
\label{eff_sys_norm}
\end{equation}
The above reasoning regarding the scaling invariance of the solution explains that when the solution exists and is unique, such normalization does not affect the convergence.

\subsection{The Main Statement}
In the robust setting described above, a more intuitive result concerning the R-CARMEL estimate and the RFF can be obtained.
\begin{theorem}
\label{tyl_mt}
Assume X = $\{\X_1,\dots\X_n\}$ are independently sampled from a continuous distribution over $\mathcal{M}_{p \times q}$ and consider the problem of minimizing $\widetilde f_{\mathcal E}(\P,\Q;X)$ over $\mathcal M_{\mathcal{E}}$, then
\begin{enumerate}[leftmargin=0.75cm]
\item if $n < \max[p/q,q/p]$, there is no unique minimum,
\item if $n> \max[p/q, q/p]$, there is a unique minimum a.s., 
\item if $n> \max[p/q, q/p]$, the normalized RFF scheme (\ref{eff_sys_norm}) converges starting from any point of $\mathcal M_{\mathcal{E}}$ to this unique minimum a.s.
\end{enumerate}
\end{theorem}
\begin{proof}
The proof can be found in Section \ref{RobustProofSec}.
\end{proof}

A few points are in place here.
\begin{rem}
Note that unlike the Gaussian case treated before, here it is natural to assume the mean to be known to get a tractable MLE program. This explains the reduction of the necessary number of samples by one in item 1) of Theorem~\ref{tyl_mt} compared to Theorem~\ref{main_g}.
\end{rem}
\begin{rem}
It is also worth noting that in the robust case the gap between items 1) and 2) consists of one point only, which distinguishes this case from the Gaussian scenario. The robust result clearly provides a sharper threshold between the mode of existence and uniqueness of the MLE and the mode where it does not at all exist.
\end{rem}

\section{Proof of Theorem~\ref{main_g}}
\label{GaussProofSec}
This section treats the Gaussian setting and utilizes a few useful concepts and techniques from commutative algebra. Therefore, for the reader's convenience we transition to a more general treatment of linear spaces, their tensor products, operators over them and a few more related notions. For this purpose, the next section introduces some additional notations.

\subsection{Additional Notations}
Abstract vector spaces are denoted by capitals $V$ and are assumed to be real Euclidean spaces, their dual spaces are denoted by $V^*$. Scalar products are denoted by $(\cdot,\cdot)$ and the corresponding norms by $\|\cdot\|$. Spectral norms of operator are denoted by $\|\cdot\|_2$. For an operator $\A$, its adjoint is denoted by $\A^*$. Note that the scalar product induces a canonical isomorphism $V\cong V^*$. If we identify $V$ with the space of columns $\mathbb R^p$ in some orthonormal basis, then $V^*$ may be identified with the space of rows $\mathbb R^p$. Then the dual basis in $V^*$ is also orthonormal and the isomorphism between $V$ and $V^*$ is given by the transpose map. For a self-adjoint operator $\A$, we naturally define its real powers and continuous functions of it via continuous functional calculus.

We may naturally identify $V\otimes V^*$ with $\End{}_{\mathbb R}(V)$ via $(\v\otimes \xi)(\u) = \xi(\u)\v$, then the scalar product on $V\otimes V^*$ induces a scalar product on $\End{}_\mathbb R(V)$ such that for any operators $\A$ and $\B,\; (\A,\B) = \Tr{\A\B^*}$. If we identify $V$ with $\mathbb R^p$, then $\End{}_\mathbb R(V)$ is identified with $\mathcal M_{p\times p}(\mathbb R)$, $\A^*$ becomes $\A^\top$, and $(\A,\B) = \Tr{\A\B^\top}$. Given two Euclidean spaces $V$ and $U$, the scalar products on $V$ and $U$ induce one on their tensor product $V\otimes U$.

For any topological space $X$, we denote its one-point compactification by $\dt X$, i.e. $\dt{X} = X \sqcup \{\infty \}$ with the base of neighborhoods of $\infty$ consisting of the sets $X\setminus K\sqcup \{\infty\}$ for all compact $K\subseteq X$. For a non-compact topological space $X$, we denote $\{\infty\}$ by $\partial X$. As an exception, for the real line $\mathbb R$, $\dt{\mathbb R}$ will be a two point compactification, i.e. $\dt{\mathbb R} = \{-\infty\}\sqcup \mathbb R\sqcup \{+\infty\}$ endowed with the usual topology making it homeomorphic to the closed unit interval. Given two sequences $\omega_n$ and $\tau_n$, we write $\omega_n \asymp \tau_n$ if $\omega_n/\tau_n \to 1$ as $n \to \infty$, while we will usually suppress the $n$ index.

In this section we shall treat the negative log-likelihoods as functions of the inverse matrices $\P^{-1}$ and $\Q^{-1}$ (as we already did while discussing the example in Section \ref{ex_case_st}). We do so to simplify calculations and note that this change does affect existence and uniqueness of the extrema in the problem at hand. Therefore, we denote
\begin{multline}
f_\mathcal{N}(\M,\P\otimes\Q;X) = \widetilde f_\mathcal{N}(\M,\P^{-1}\otimes\Q^{-1};X) \\ = \frac{1}{n} \sum_{i=1}^n \Tr{\P(\X_i -\M) \Q(\X_i -\M)^\top} - \ln|\P\otimes\Q|,
\end{multline}
\begin{equation}
g_{\mathcal N}(\P\otimes\Q;X) = f_{\mathcal N}(0,\P\otimes\Q;X) = \frac{1}{n} \sum_{i=1}^n \Tr{\P\X_i\Q\X_i^\top} - \ln|\P\otimes\Q|.
\end{equation}

\subsection{Metric over $\mathcal P(V)$}\label{metricSubSec}
We endow the cone of the positive definite operators over $V,\;\mathcal P(V)$, with a Riemannian metric, whose geodesic connecting any two points $\P,\R \in \mathcal P(V)$ is given by
\begin{equation}
\gamma_t(\P,\R) = \P^{\frac{1}{2}} \(\P^{-\frac{1}{2}} \R \P^{-\frac{1}{2}}\)^t \P^{\frac{1}{2}},\quad 0\leqslant t \leqslant 1.
\label{metric_def}
\end{equation}
Due to the limited space, we omit a discussion of this metric and its properties. For more details on the relation of this metric to the MLE problems considered here, consult \cite{rapcsak1991geodesic, wiesel2012geodesic, wiesel2012unified, wiesel2011regularized} and references therein. If we allow $t$ to run over $\mathbb{R}$, then we call the obtained curve an extended geodesic curve.

\begin{fact}\label{remark:gMetricDet}
A direct computation shows that the Riemannian metric we have just introduced is invariant under inversion. In addition, we note that the log-determinant function is a linear function of $t$ on the geodesic curves $\gamma_t(\P,\R)$.
\end{fact}

\begin{lemma}\label{lemma:PowerOfPos}
Let $V$ be a vector space, $\x\in V$, $\S \in \mathcal P(V)$ and
\begin{equation}
\varphi_\x(t) = (\S^t \x,\x),\;\; t \in \mathbb{R},
\end{equation}
then its second derivative reads as
\begin{equation}
\varphi''_\x(t) = (\ln^2 (\S)\, \S^t \x,\x) = \norm{\ln (\S)\, \S^{\frac{t}{2}}\x}^2 \geqslant 0,
\label{n_eq_1}
\end{equation}
in particular, $\varphi_\x(t)$ is convex. In addition, the following are equivalent:
\begin{enumerate}
\item $\varphi''_\x(t) = 0\;\;$ for some $t \in \mathbb R$,
\item $\varphi''_\x(t) \equiv 0\;\;$ for all $t\in \mathbb R$,
\item $\S \x = \x$.
\end{enumerate}
Therefore, if $\varphi_\x$ is linear in an open neighborhood of some $t_0$, then $\varphi_\x$ is constant on the whole $\mathbb R$.
\end{lemma}
\begin{proof}
Since $(\S^t)' = \ln (\S)\, \S^t$ and $\ln (\S)$ commutes with any power of $\S$, we get (\ref{n_eq_1}) and the convexity follows.

Note that 2) implies 1) and, therefore, it is enough to show $1)\rightarrow 3)\rightarrow 2)$. If $\S\x = \x$, then for any real $t$, $\S^t\x = \x$. Hence, $\varphi_\x$ is constant on $\mathbb R$ and we get $3)\rightarrow 2)$. Now let $\varphi''_x(t) = 0$, then $\ln (\S)\, \S^{\frac{t}{2}}\x = 0$. Since $\S$ is invertible, so is $\S^{\frac{t}{2}}$, and $\ln (\S) \x = 0$ or, equivalently, $\x$ is an eigenvector of $\S$ with the eigenvalue $1$. Finally, the last claim follows from the fact that $\varphi_\x(t) = \norm{\S^{\frac{t}{2}}\x}^2 \geqslant 0$ and the only linear nonnegative function is a constant function.
\end{proof}

\begin{lemma}\label{lemma:ScalarPrdConvex}
Let $V$ be a vector space, $\v\in V$, and $\P, \R \in \mathcal P(V)$, then
\begin{equation}
\omega_\v(t) = (\v,\gamma_t(\P,\R)\v)
\end{equation}
is convex and the following are equivalent:
\begin{enumerate}
\item $\omega_\v$ is linear on some open subset of $\mathbb R$,
\item $\omega_\v$ is constant on the whole $\mathbb R$,
\item $\P\v = \R\v$.
\end{enumerate}
\end{lemma}
\begin{proof}
It is an immediate corollary of Lemma~\ref{lemma:PowerOfPos} if we set $\S = \P^{-\frac{1}{2}} \R \P^{-\frac{1}{2}}$ and $\x = \P^{\frac{1}{2}}\v$.
\end{proof}

\subsection{Convexity of $\mathcal M_\mathcal N$ and $g_\mathcal N$}
Let $V$ and $U$ be vector spaces, then their tensor product naturally induces a map
\begin{equation}
\otimes \colon \mathcal P(V)\times \mathcal P(U)\to \mathcal P(V\otimes U),
\label{tens_m_def}
\end{equation}
sending a pair $(\P,\Q)$ to the product $\P\otimes \Q$. We denote the image of this map by $\mathcal M_\mathcal N$. The identification
\begin{equation}
\mathcal M_\mathcal{N} \cong \{(\P,\Q) \mid \norm{\P}_2 = 1\}\subset \mathcal P(V)\times \mathcal P(U)
\label{M_def}
\end{equation}
provides $\mathcal M_\mathcal N$ with a structure of a smooth manifold. Intuitively, this amounts to saying that fixing the norm of the first component of a KP resolves the scaling ambiguity and provides a bijective correspondence between the factors and their products. Note that the normalization in (\ref{M_def}) is chosen arbitrarily, and the specific choice does not affect the existence and uniqueness results. In addition, we have
\begin{lemma}
The manifold $\mathcal M_\mathcal N\subset \mathcal P(V\otimes U)$ is convex w.r.t. the geodesic metric defined in Section \ref{metricSubSec}.
\end{lemma}
\begin{proof}
Since $\P\otimes \Q = \P\otimes \I \cdot \I\otimes \Q$ and $\P\otimes \I$ commutes with $\I\otimes \Q$, we have
\begin{equation}
\gamma_{t}(\P\otimes \Q,\R\otimes \T) = \gamma_t(\P,\R)\otimes \gamma_t(\Q,\T),
\end{equation}
where the right-hand side is in $\mathcal M_\mathcal N$, thus we are done.
\end{proof}

Below we also make use of the following simple
\begin{fact}
For any distinct $\P,\R \in \mathcal P(V),\; \gamma_t(\P,\Q) \to \infty$ as $t \to \pm \infty$.
\end{fact}

\begin{lemma}\label{lemma:hatgProp}
Let $V$ be a vector space, $X\subset V$ - a fixed finite subset, $\P\in\mathcal P(V)$ and
\begin{equation}
g(\P;X) = \frac{1}{|X|}\sum_{\x\in X} (\P\x,\x) - \ln|\P|,
\end{equation}
then
\begin{enumerate}
\item $g$ is convex w.r.t. the Riemannian metric (\ref{metric_def}),
\item $g$ is linear on $\gamma_t(\P,\R)$ for some $\P, \R \in \mathcal P(V)$ if and only if $\P\x=\R\x$ for all $\x\in X$,
\item if $g$ has two minima $\P \neq \R$ in $\mathcal P(V)$, then the whole extended geodesic $\gamma_{t}(\P,\R),\; t\in \mathbb R$ consists of different minima,
\item let $U$ be another vector space, $\mathcal M_\mathcal{N} \subset \mathcal P(V\otimes U)$ as before, and $\dt g_\mathcal{N} \colon \dt{\mathcal M}_\mathcal{N}\to \dt{\mathbb R}$ extends $g_\mathcal{N}$ such that $\dt g_\mathcal{N}(\infty;X) = +\infty$, then $\dt g_\mathcal{N}$ is continuous if and only if $g_\mathcal{N}$ has a unique minimum.
\end{enumerate}
\end{lemma}
\begin{proof}
1) $\ln\det \P$ is linear on geodesics and, thus, convex by Fact \ref{remark:gMetricDet}. Lemma~\ref{lemma:ScalarPrdConvex} implies that each $(\P\x,\x)$ is convex, therefore, so is $g(\P;X)$.

2) This follows from the $3) \to 2)$ implication of Lemma~\ref{lemma:ScalarPrdConvex}.

3) The convexity implies that the restriction of $g$ onto $\gamma_t(\P,\R),\; 0\leqslant t \leqslant 1$, is constant and, therefore, linear. Now the $1) \to 2)$ implication of Lemma~\ref{lemma:ScalarPrdConvex} finishes the proof.

4) Let us show the sufficiency of the condition. Indeed, if $\dt g_{\mathcal N}$ is continuous then it achieves a minimum at some interior point (that is the existence). If such minimum is not unique, then by 3), $\dt g_{\mathcal N}$ must be constant on the whole extended geodesic and cannot be continuous when approaching the boundary, since $\dt g_{\mathcal N}(\infty;X) = +\infty$.

We proceed to the necessity. Let $\S_0 \in \mathcal M_\mathcal N$ be the unique minimum and $\nu = g_{\mathcal N}(\S_0;X)$. We denote by $\mathcal{T}_{\S_0} \mathcal M_\mathcal N$ the tangent space to our manifold at the point $\S_0$, and choose a tangent vector $\R \in \mathcal{T}_{\S_0} \mathcal M_\mathcal N$. Let $\gamma_{t}(\R)$ be the geodesic starting at $\S_0$ in direction $\R$,
\begin{equation}
\gamma_0(\R) = \S_0 \quad \text{and}\quad \gamma_0'(\R) = \R.
\end{equation}
The explicit formula for $\gamma_t$ reads as
\begin{equation}
\gamma_t(\R) = \C_0e^{t \C_0^{-1}\R \C_0^{-1}}\C_0,
\end{equation}
where $\C_0 = \S_0^{\frac{1}{2}}$. In particular, $\gamma_{\lambda t}(\R) = \gamma_t(\lambda \R)$ for any $\lambda > 0$. Set $\delta_t(\R) = g_{\mathcal N}(\gamma_t(\R);X)$. We claim that
\begin{equation}
\sigma = \min_{|\R| = 1} \delta'_1(\R) > 0.
\end{equation}
Indeed, suppose $\delta'_1(\R_0) = 0$ for some $\R_0$, then since $\delta_t(\R_0)$ is convex and $t = 0$ is the minimum, $\delta_t(\R_0)$ is constant for $0 \leqslant t \leqslant 1$. Thus, the minimum is not unique, which is a contradiction.

Denote by $B^0_t$ the centered open ball of radius $t$ in $\mathcal{T}_{\S_0}\mathcal M_\mathcal N$ and $B_t$ - its closure, then $\gamma_1(B^0_t)$ is a family of open neighborhoods of $\S_0$ with compact closure $K_t = \gamma_1(B_t)$. Thus, we need to show that 
\begin{equation}
\inf_{\P \otimes \Q \notin K_t} g_{\mathcal N}(\P \otimes \Q;X) \to +\infty, \text{ as } t\to +\infty.
\end{equation}
Indeed,
\begin{multline}
\inf_{\P\otimes \Q\notin K_t}g_{\mathcal N}(\P \otimes \Q;X) = \inf_{\substack{\R \in \mathcal{T}_{\S_0}\mathcal M_\mathcal N, \\ \norm{\R}\geqslant t}} g_{\mathcal N}(\gamma_1(\R);X) \\ \geqslant g_{\mathcal N}\(\gamma_{\norm{\R}}\(\frac{\R}{\norm{\R}}\);X\) \geqslant \sigma (\norm{\R} - 1) + \nu \geqslant \sigma(t-1) + \nu, 
\end{multline}
where in the last line $\R \in \mathcal{T}_{\S_0}\mathcal M_\mathcal N$ is any matrix of norm at least $t$.
\end{proof}

\subsection{The Set of ``Bad'' Samples}
Depending on the set $X$, $g_\mathcal N$ may happen to be not strictly convex on $\mathcal M_\mathcal N$, or equivalently, $\dt g_\mathcal N$ is not necessarily continuous. In this section we discuss when this situation occurs and compute the measure of the set of samples making $\dt g_\mathcal N$ discontinuous.

For a vector space $V$ of dimension $p$ and natural numbers $d, s \in \mathbb{N}$, define
\begin{equation}
\Gr_{d\,s}(V) = \{(\v_1,\ldots,\v_s) \in V^s\mid \dim \langle \v_1,\ldots,\v_s\rangle = d \}\subseteq V^s
\end{equation}
to be the set of all $s$-tuples of vectors in $V$, spanning subspaces of dimension $d$. $\Gr_{d\,s}(V)$ is a smooth manifold of dimension $(p + s - d) d$. Note also that $\Gr_{d\,d}(V)\subseteq V^d$ is an open subset, moreover, if we represent $V^d = V\otimes \mathbb R^d$, we get an action of $\GL_d(\mathbb R)$ on $V^d$, which restricts correctly onto $\Gr_{d\,d}(V)$ and is free.

Before giving a precise statement about what Diagram \ref{diag_def} displays, let us provide an intuitive explanation of it. The operator analog of a  $p\times q$ matrix is an element of $\Hom(U,V)$, thus our $n$ matrix measurements in $X$ together represent an element of $\Hom(U,V)^n$. We next take $d$ linearly independent vectors in a $q$ dimensional $U$ and apply all the elements of our $n$-tuple to them - this provides us with $dn$ vectors in a $p$ dimensional $V$. If we now take all the pairs of such $n$-tuples of operators and $d$-tuples of vectors in $U$, and consider the described action of the former on the latter, which we call $\Psi$, we get the first line of the diagram at hand. Let us now consider all the sets of $dn$ vectors inside $V$ spanning subspaces of dimension $r$ and take their preimage under $\Psi$. We get a subset $\mathcal Z_{d\,n\,r}$ of $\Hom(U,V)^n \times \Gr_{d\,d}(U)$ depicted in the diagram. Finally, the leftmost arrow $\pi$ denotes the projection of this set onto $\Hom(U,V)^n$. This informal description is made precise by the following

\begin{figure*}[!t]
\begin{equation*}
\xymatrix{
{}&{\mathcal A_{n\,d} = \Hom(U,V)^n\times \Gr_{d\,d}(U)}\ar[r]^-{\Psi}&{V^{dn}}\\
{\Hom(U,V)^n}&{\mathcal Z_{d\,n\,r}}\ar[r]\ar@{}[u]|{\bigcup}\ar[l]_-{\pi}&{\Gr_{r\,dn}(V)}\ar@{}[u]|{\bigcup}\\
}
\end{equation*}
\caption{Diagram from Definition \ref{def:diagram}.}
\label{diag_def}
\end{figure*}

\begin{definition}\label{def:diagram}
Let $V$ and $U$ be vector spaces of dimensions $p$ and $q$, respectively, and $n, d, r \in \mathbb{N}$ be such that $d\leqslant q$. Consider Diagram \ref{diag_def}, where $\Psi$ is defined as
\begin{gather}
\Psi : \Hom(U,V)^n \times \Gr_{d\,d}(U) \to V^{dn}, \\
((\varphi_1,\ldots,\varphi_n),(\u_1,\ldots,\u_d)) \mapsto (\varphi_i \u_j).
\end{gather}
Identify $V^{dn} = V^n\otimes \mathbb R^d$, then $\Psi$ reads as
\begin{equation}
\Psi((\varphi_1,\ldots,\varphi_n),(\u_1,\ldots,\u_d)) = \varphi_i \u_j \otimes \e_j,
\end{equation}
and $\C \in\GL_d(\mathbb R)$ acts by $\I \otimes \C$. Define $\mathcal Z_{d\,n\,r} = \Psi^{-1}\(\Gr_{r\,dn}(V)\)$ and $\pi$ to be the restriction of the projection along $\Gr_{d\,d}(U)$.
\end{definition}

\begin{lemma}\label{lemma:surjectivity}
With the notations of Definition \ref{def:diagram}, for each $(\u_1,\dots,\u_d) \in \Gr_{d\,d}(U)$ the map
\begin{gather}
\Omega : \Hom(U,V)^n  \to V^{dn}, \\
(\varphi_1,\dots,\varphi_n) \mapsto (\varphi_i \u_j)
\end{gather}
is surjective.
\end{lemma}
\begin{proof}
$\Omega$ is a direct sum of $n$ maps
\begin{gather}
\omega \colon \Hom(U,V)\to V^d, \\ 
\varphi \mapsto (\varphi \u_1,\dots, \varphi \u_d).
\end{gather}
Now we choose bases in $V$ and $U$ such that $V=\mathbb R^p$ and $U=\mathbb R^q$. Then $\Hom(U,V) = \mathcal M_{p\times q}$ and $V^d = \mathcal M_{p\times d}$. Let $\U = \[\u_1,\ldots,\u_d\]$, then $\omega$ reads as
\begin{gather}
\omega\colon \mathcal M_{p\times q}(\mathbb R)\to \mathcal M_{p\times d}(\mathbb R), \\ 
\X\mapsto \X\U.
\end{gather}
Since $d\leqslant q$ and the columns of $\U$ are linearly independent, $\U$ is of full rank. Hence, the map $\omega$ is surjective.

\end{proof}

\begin{lemma}\label{lemma:diagram}
With the notations of Definition \ref{def:diagram},
\begin{enumerate}
\item $\Psi$ is surjective and $d\Psi$ is surjective at each point,
\item $\mathcal Z_{d\,n\,r}$ is a smooth manifold with
\begin{equation}
\dim \mathcal Z_{d\,n\,r} = pqn + qd - \(p - r\)\(dn - r\),
\end{equation}
\item a non empty fiber of $\pi$ has dimension at least $d^2$.
\end{enumerate}
\end{lemma}
\begin{proof}
1) To show $\Psi$ is surjective, let $(\u_1,\dots,\u_d) \in \Gr_{d\,d}(U)$. Now it is enough to show that
\begin{equation}
\Psi(\cdot,(\u_1,\dots,\u_d))\colon \Hom(U,V)^n \to V^{dn}
\end{equation}
is surjective, which follows from Lemma~\ref{lemma:surjectivity}. We proceed to the surjectivity of $d\Psi$ at any point $(\{\varphi_i\},\{\u_j\})$. Since $\Gr_{d\,d}(U) \subseteq U^d$ is open, identify the tangent space of $\Gr_{d\,d}(U)$ at $\{\u_j\}$ with $U^d$, then
\begin{gather}
d_{(\{\varphi_i\},\{\u_j\})}\Psi \colon \Hom(U,V)^n\times U^d \to V^{dn},\\
(\{\varphi_i^*\},\{\u_j^*\})\mapsto (\varphi_i^* \u_j + \varphi_i \u_j^*).
\end{gather}
Taking $\u_j^* = 0,\; j=1,\dots,d$, we get $(\{\varphi_i^*\},\{\bm{0}\}) \mapsto (\varphi_i^* \u_j)$ and the result follows from Lemma~\ref{lemma:surjectivity}.

2) The fact that $\mathcal Z_{d\,n\,r}$ is a smooth manifold follows from 1) and the Implicit Function Theorem. A direct computation yields
\begin{multline}
\dim \mathcal Z_{d\,n\,r} = \dim \mathcal A_{n\,d} + \dim \Gr_{r\,dn}(V) - \dim V^{dn} \\ = pqn + qd+(p+dn-r)r - pdn = pqn + qd - (p - r)(dn - r).
\end{multline}

3) $\GL_d(\mathbb R)$ acts freely on the fibers of $\pi$, so the dimension of a fiber is at least $\dim \GL_d(\mathbb R) = d^2$.
\end{proof}

\begin{corollary}\label{cor:LebZero}
With the notations of Lemma~\ref{lemma:diagram}, if $n > p/q + q/p$ and $r\leqslant dp/q$, then $\pi(\mathcal Z_{d\,n\,r})$ is a set of Lebesgue measure zero.
\end{corollary}
\begin{proof}
By Sard's theorem, it is enough to show that the image of $\pi$ consists of critical values only. So we need to show that
\begin{equation}
\rk{\pi} < \dim \Hom(U,V)^n.
\end{equation}
Since
\begin{equation}
\rk{\pi}\leqslant \dim \pi(\mathcal Z_{d\,n\,r}) - d^2,
\end{equation}
by Lemma~\ref{lemma:diagram} item 3), it is enough to prove that
\begin{equation}
\dim \pi(\mathcal Z_{d\,n\,r}) - d^2 < \dim \Hom(U,V)^n.
\end{equation}
By Lemma~\ref{lemma:diagram} item 2), the latter is equivalent to
\begin{equation}
pqn + qd - (p - r)(dn - r) < pqn,
\end{equation}
therefore, we need
\begin{equation}
\frac{q-d}{p-r} + \frac{r}{d} < n,\;\; \forall\; 1\leqslant d \leqslant q \text{ and } 0\leqslant r \leqslant d\frac{p}{q}.
\end{equation}
Differentiate the left hand-side w.r.t. $r$ to see that it is a strictly increasing function of $r$, thus, it is enough to demonstrate the inequality for $r = dp/q$, which is
\begin{equation}
\frac{q - d}{p - d\frac{p}{q}} + \frac{p}{q} = \frac{q}{p} + \frac{p}{q}< n,
\end{equation}
and holds by the assumption.
\end{proof}

\begin{corollary}\label{cor:enoughSamples}
Let $V$ and $U$ be vector spaces of dimensions $p$ and $q$, respectively, and $X$ - a finite mutually continuous family of random operators from $U$ to $V$ such that
\begin{equation}
|X| > \frac{p}{q} + \frac{q}{p},
\end{equation}
then for each random subspace $E\subseteq V$ we have
\begin{equation}
\frac{\dim \sum_{\X\in X}\X E}{\dim E} > \frac{p}{q} \quad\text{a.s.}
\end{equation}
\end{corollary}
\begin{proof}
Note that $X$ is distributed over $\Hom(U,V)^n$, then
\begin{equation}
\left\{X \;\left|\; \exists E\subseteq U:\;\frac{\dim \sum_{\X\in X}\X E}{\dim E} \leqslant  \frac{p}{q} \right\} \right.= \bigcup_{\substack{1\leqslant d\leqslant q\\ 0\leqslant r \leqslant d \frac{p}{q}}}\pi(\mathcal Z_{d\,n\,r}),
\end{equation}
and the result follows from Corollary~\ref{cor:LebZero}.
\end{proof}

Note that in Corollary~\ref{cor:enoughSamples}, we do not care whether $p\geqslant q$ or $q\geqslant p$. When $p\geqslant q$, the inequality may fail if $\X E$ is not big enough compared to $E$, and when $q\geqslant p$, it may fail if $E$ belongs to the kernels of all samples $\X$.

\subsection{Flags}
In the proof of the main theorem (see Theorem~\ref{theorem:continuity} below) we will analyze the behavior of $g_\mathcal N(\P\otimes\Q;X)$ when $\P\otimes \Q$ tends to $\infty$. There are many ways $\P\otimes \Q$ may tend to the boundary and in order to classify all possibilities we utilize the flag machinery introduced next.

\begin{definition}\label{def:Gflag}
Let $U$ be a vector space, then a flag $\mathcal F$ of length $s$ in $U$ is an ascending sequence of  proper subspaces of $U$
\begin{equation}
\mathcal F = \{0 = U_0 \subsetneq U_1 \subsetneq \ldots \subsetneq U_s \mid U_s\subseteq U\}.
\end{equation}
\end{definition}
The flag is called non-trivial if $0 \subsetneq U_1 \subsetneq U$. A subsequence of $\mathcal{F}$ is called a subflag. Let $V$ be another vector space and $\mathcal G = \{V_i\}$ be a flag of length $s$ in $V$. Let $\zeta \colon U\to V$ be a linear map with $\zeta U_i\subseteq V_i$ for each $i\leqslant s$, then we write $\zeta \mathcal F\subseteq \mathcal G$.
In addition, for all $0\leqslant i,j\leqslant s$ define
\begin{equation}
\Pi(\mathcal F,\mathcal G)_{i\,j} = q (\dim V_j - \dim V_i) - p (\dim U_j - \dim U_i).
\end{equation}
If $\r = \{r_1>\ldots >r_s>0\}$ is a vector of strictly decreasing real numbers, we define
\begin{equation}
C(\mathcal F,\mathcal G,\r) = \sum_{i=1}^s r_i \Pi(\mathcal F,\mathcal G)_{i-1\,i},
\end{equation}
and note that
\begin{equation}
\Pi(\mathcal F,\mathcal G)_{i\,j} + \Pi(\mathcal F,\mathcal G)_{j\,k} = \Pi(\mathcal F,\mathcal G)_{i\,k}.
\end{equation}

\begin{lemma}\label{lemma:flagCoeff}
Let $V, U$ and $\mathcal G, \mathcal F$ be as above. If $\Pi(\mathcal F,\mathcal G)_{0\,i} > 0$ for $i=1,\dots,s$, then there exist subflags $\mathcal F'\subseteq \mathcal F,\;\mathcal G'\subseteq \mathcal G$, and a subsequence $\r'\subseteq \r$ - all of length $s'$, such that
\begin{equation}
\Pi(\mathcal F',\mathcal G')_{0\,1} > 0,\quad \Pi(\mathcal F',\mathcal G')_{i-1\,i} \geqslant 0,\;\; i=1,\dots, s',
\end{equation}
and $C(\mathcal F,\mathcal G,\r)\geqslant C(\mathcal F',\mathcal G',\r')$. In particular, $C(\mathcal F,\mathcal G,\r) > 0$.
\end{lemma}
\begin{proof}
The proof is by induction on $s$ and the base $s = 1$ is the hypothesis of the lemma. Now let $s>1$ and assume the claim fails for $\mathcal F$ and $\mathcal G$, then for some $t\leqslant s$
\begin{equation}
\Pi(\mathcal F,\mathcal G)_{i-1\,i}\geqslant 0\quad \forall i<t, \quad\text{and}\quad \Pi(\mathcal F,\mathcal G)_{t-1\,t} < 0.
\end{equation}
Construct $\mathcal F'$ and $\mathcal G'$ from $\mathcal F$ and $\mathcal G$ by excluding $U_{t-1}$ and $V_{t-1}$, respectively, and $\r'$ from $\r$ by excluding $r_{t-1}$, then
\begin{multline}
C(\mathcal F,\mathcal G,\r) 
= \sum_{i \neq t-1,t} r_i \Pi(\mathcal F,\mathcal G)_{i-1\,i}+ r_{t-1} \Pi(\mathcal F,\mathcal G)_{t-2\,t-1} + r_{t}\Pi(\mathcal F,\mathcal G)_{t-1\,t} \\
\geqslant \sum_{i \neq t-1,t} r_i \Pi(\mathcal F,\mathcal G)_{i-1\,i} +  r_{t}\( \Pi(\mathcal F,\mathcal G)_{t-2\,t-1} + \Pi(\mathcal F, \mathcal G)_{t-1\,t}\)
= C(\mathcal F',\mathcal G', \r').
\end{multline}
\end{proof}

\subsection{The Known Mean Case}
In this section, we prove our main result assuming the mean to be known. For this purpose we need a few auxiliary results.
\begin{lemma}\label{lemma:lambdaMu}
Let $\lambda,\mu$ and $\gamma$ be positive sequences such that $\lambda \mu = O(\ln \gamma)$ and $\gamma\to +\infty$, then
\begin{equation}
\ln \lambda^{-1} \geqslant  \ln \mu + o(\ln \gamma).
\end{equation}
\end{lemma}
\begin{proof}
$\lambda \mu = \alpha \ln\gamma$, with $\alpha \leqslant \kappa$ for some constant $\kappa$. Taking logarithms yields
\begin{equation}
\ln \lambda + \ln \mu = \ln \,\ln \gamma + \ln \alpha,
\end{equation}
hence,
\begin{equation}
\ln \lambda^{-1} = \ln \mu -\ln\,\ln\gamma - \ln\alpha \geqslant \ln\mu - \ln\,\ln \gamma - \ln \kappa = \ln\mu + o(\ln \gamma).
\end{equation}
\end{proof}

\begin{theorem}\label{theorem:continuity}
Let $V$ and $U$ be vector spaces of dimensions $p$ and $q$, respectively, $\mathcal M_\mathcal{N} \subset \mathcal{P}(V \otimes U)$ as in (\ref{tens_m_def}) and $X\subset V\otimes U$ - a finite mutually continuous family of random vectors such that $|X|>  p/q + q/p$, then $g_\mathcal{N} \colon \mathcal M_\mathcal{N}\to \mathbb R$ extends to a continuous function $\dt g_\mathcal{N} \colon \dt {\mathcal M}_\mathcal{N} \to \dt {\mathbb R}$ via $\dt g_\mathcal{N}(\infty;X) =+\infty$. In particular, there exists a unique minimum of $g_\mathcal{N}$ over $\mathcal M_\mathcal{N}$.
\end{theorem}

\begin{rem}
Note that the statement allows the members of $X$ to be statistically dependent and does not require identical distribution. This generality is necessary when we treat the case of manually empirically centered samples below and makes application of Theorem~\ref{theorem:continuity} possible without additional adjustments.
\end{rem}

\begin{proof}
By Lemma~\ref{lemma:hatgProp} item 4) it is enough to show that $g_{\mathcal N}(\P,\Q;X)\to +\infty$ as $(\P,\Q) \to \infty$. Suppose on the contrary, that there exists a sequence $(\P,\Q) \to \infty$ (we omit the $j$ indexing in $\{(\P_j,\Q_j)\}_j$ to simply notations) such that $g_{\mathcal N}(\P,\Q;X) \leqslant \kappa$ for some constant $\kappa$. Rewrite $g_{\mathcal N}$ as
\begin{equation}
g_{\mathcal N}(\P,\Q;X) = \sum_{\X\in X} \varphi_\X(\P,\Q) + \psi(\P,\Q) = \varphi(\P,\Q) + \psi(\P,\Q),
\label{g_sum}
\end{equation}
where
\begin{equation}
\varphi_\X(\P,\Q) = \frac{1}{|X|}(\P\X\Q,\X),\;\; \psi(\P,\Q) = -\ln|\P\otimes \Q|.
\end{equation}
Recall that $\mathcal M_\mathcal N$ can be identified with
\begin{equation}
\mathcal M_\mathcal N \cong \{(\P,\Q)\mid \|\P\|_2 = 1 \}\subset \mathcal P(V)\times \mathcal P(U).
\end{equation}
If $\norm{\Q}_2$ is bounded, then the sequence $(\P,\Q)$ tends to a singular pair (at least one of the matrices tends to a singular limit). In this case, $\varphi(\P,\Q) \asymp O(1)$ and $\psi(\P,\Q) \to +\infty$.

Now assume $\norm{\Q}_2 \to +\infty$, the only problem here is that $-p\,\ln|\Q|$ may tend to $-\infty$. We should show that we can compensate for this with the other summands. Let $\bm\sigma^\Q$ be the spectrum of $\Q$, then it can be partitioned as $\bm\sigma^\Q = \bm\sigma_\infty^\Q \sqcup \bm\sigma_m^\Q$ such that
\begin{itemize}
\item for each $\mu\in \bm\sigma_\infty^\Q,\; \ln \mu \asymp r_\mu\ln \norm{\Q}_2$, where $r_\mu > 0$ is constant,
\item for each $\mu \in \bm\sigma_m^\Q,\; \ln\mu = o(\ln \norm{\Q}_2)$.
\end{itemize}
Order the elements of $\bm\sigma_\infty^\Q$ by their rate of convergence
\begin{equation}
\bm\sigma_\infty^\Q = \bm\sigma_1^\Q\sqcup \ldots\sqcup \bm\sigma_s^\Q,
\end{equation}
where
\begin{itemize}
\item for each $\mu \in \bm\sigma_1^\Q,\; \mu \asymp \norm{\Q}_2$,
\item for each $\mu, \mu' \in \bm\sigma_i^\Q,\; \lim \mu/\mu'$ is a non-zero constant,
\item for any $i$, if $\mu_i\in \bm\sigma_i^\Q$, then $\mu_{i+1} = o(\mu_i)$.
\end{itemize}
For a fixed $i$, let $\{\bar K_i\}$ be the sequence of random subspaces of $U$ generated by the eigenvectors corresponding to $\bm\sigma_i$, and $K_i$ be the limit of $\{\bar K_i\}$ (it exists after passing to an appropriate subsequence, if needed). Now $U_k = \oplus_{j\leqslant k} K_j$ form a non-trivial random flag of length $s$ in $U$,
\begin{equation}
\mathcal F = \{0 = U_0\subseteq U_1\subseteq U_2\subseteq \ldots\subseteq U_s\mid U_s\subseteq U\}.
\end{equation}
Let $\bm\sigma^\P$ be the spectrum of $\P$ and set
\begin{equation}
\bm\sigma^\P_i = \{\lambda\in \bm\sigma^\P\mid \lambda \mu_i = O(\ln\norm{\Q}_2),\; \text{for}\;\; \mu_i \in \bm\sigma_i^\Q\}.
\end{equation}
By the definition of $\bm\sigma_i^\Q,\; \bm\sigma_i^\P$ does not depend on the choice of $\mu_i \in \bm\sigma_i^\Q$. Let $\{\bar L_i\}$ be the sequence of random subspaces of $V$ generated by the eigenvectors corresponding to $\bm\sigma^\P_i$ and $L_i$ be the limit of $\{\bar L_i\}$ (here again, it exists after passing to an appropriate subsequence).
Denote $V_k = \oplus_{j\leqslant k} L_j$ and define
\begin{equation}
\mathcal G = \{0 = V_0 \subseteq V_1\subseteq \ldots \subseteq V_s\mid V_s\subseteq V\},
\end{equation}
which is a flag of length $s$ in $V$. We denote the orthogonal projector in $V\otimes U$ onto $L_j\otimes K_i$ by $\pi_{ji}$.

Now there are two possibilities:
\begin{itemize}
[leftmargin=*]
\item {\bf{$\exists \X\in X: \X\mathcal F\not\subseteq \mathcal G$}}. Let $\X U_i\not\subseteq V_i$, i.e. there is some $j$ such that $\ln \norm{\Q}_2 = o(\lambda_i \mu_j)$ and $\pi_{ij}(\X) \neq 0$, then
\begin{equation}
\varphi(\P,\Q) \geqslant \varphi_\X(\P,\Q) = \frac{1}{|X|}(\P\X\Q,\X)\geqslant \frac{1}{|X|}\lambda_i\mu_j|\pi_{ij}(\X)|^2,
\end{equation}
and $\psi(\P,\Q) = O(\ln\norm{\Q}_2) = o(\lambda_i\mu_j)$, hence, $g_{\mathcal N}(\P,\Q;X) \to +\infty$.

\item {\bf{$\forall \X\in X : \X\mathcal F\subseteq \mathcal G$}}.
Since $\varphi(\P,\Q) \not\to -\infty$, we can ignore this summand when considering the asymptotic behavior. Compute $\psi(\P,\Q) = - q\, \ln|\P| - p\, \ln|\Q|$ explicitly,
\begin{equation}
-q\,\ln|\P| = q\sum_{i=1}^s (\dim V_i - \dim V_{i-1})\ln \lambda_i^{-1} - q\ln\determ \P|_{V_s^\perp},
\end{equation}
\begin{equation}
-p\,\ln|\Q| = -p \sum_{i=1}^s(\dim U_i - \dim U_{i-1})\ln \mu_i + o(\ln\norm{\Q}_2),
\end{equation}
where $\lambda_i\in \bm\sigma^\P_i$ and $\mu_i\in \bm\sigma^\Q_i$. Note that $-q\,\ln\determ \P|_{V_s^\perp} \not\to -\infty$, therefore, we may drop this summand. Since $\mu_i\asymp r_i \ln\norm{\Q}_2$, by Lemma~\ref{lemma:lambdaMu} we obtain
\begin{equation}
-\ln|\P\otimes \Q| \gtrsim C(\mathcal F,\mathcal G, \r)\ln \norm{\Q}_2 + o(\ln \norm{\Q}_2),
\end{equation}
thus, it is enough to prove that the coefficient $C(\mathcal F,\mathcal G, \r)$ is positive. This would follow from Lemma~\ref{lemma:flagCoeff} if we prove that $\Pi(\mathcal F,\mathcal G)_{0\,i} > 0$ for $1\leqslant i \leqslant s$. Indeed,
\begin{equation}
\Pi(\mathcal F,\mathcal G)_{0\,i} = q \dim V_i - p \dim U_i
= q\dim U_i \(\frac{\dim X U_i}{\dim U_i} - \frac{p}{q}\).
\end{equation}
Since $\mathcal F$ is non-trivial, $\dim U_i \neq 0$ for $i\geqslant 1$. $|X|> p/q+q/p$, thus, due to Corollary~\ref{cor:enoughSamples} the expression in brackets is a.s. positive. This finishes the proof.
\end{itemize}
\end{proof}

The proof we have just presented may be complicated to grasp due a large amount of new notations and technical details, therefore, we now explain it in an informal way. Using the same notations, let us describe the main point of using flags. Choose bases in $U$ and $V$ respecting the subspaces $K_i$ and $L_j$. In these bases all the samples $\X\in X$ has $s$ blocks of rows and $s$ block of columns corresponding to  $L_i$ and $K_i$, respectively. Hence, each sample consists of $s^2$ blocks. Now we can easily count the contributions of the blocks to the asymptotic of $g_{\mathcal N}(\P,\Q;X)$.

The contributed speed of the $(i,j)$-th block of any $\X$ is $\lambda_i\mu_j$, up to a scalar depending on $\X$. In order to determine the asymptotic behavior of $g_{\mathcal N}(\P,\Q;X)$ written as in (\ref{g_sum}), we need to compare the negative impact of $\psi(\P,\Q)$ with the positive one of $\varphi(\P,\Q)$. The highest rate negative summand appearing in $\psi(\P,\Q)$ decreases with the rate of at most $\ln \norm{\Q}_2$ up to a fixed scalar. If $\lambda_i\mu_j$ tends to infinity faster than $\ln \norm{\Q}_2$, then $g_{\mathcal N}(\P,\Q;X)$ would tend to $+\infty$.

The problem occurs if all the blocks corresponding to those $\lambda_i\mu_j$ growing faster than $\ln \norm{\Q}_2$ are zero for all $\X \in X$. Let us note that if the $(i,j)$-th block is zero for all $\X$, then all the blocks with smaller $j$ and higher $i$ (to the left and down of our block) have higher speed and, hence, must be zero (otherwise we are in the first situation). This precisely means that all the samples $\X\in X$ are block upper triangular, i.e. they map flag $\mathcal{F}$ into $\mathcal{G}$.

Now we just use these observations together with Lemmas \ref{lemma:flagCoeff} and \ref{lemma:lambdaMu} to explicitly calculate the leading asymptotic term of $\psi(\P,\Q)$, which thanks to Corollary~\ref{cor:enoughSamples} grows to $+\infty$ and, therefore, implies the desired.

\subsection{The Unknown Mean Case}
\label{main_th_pr_sec}
Let $V$ be a vector space, $X = \{\x_1,\ldots,\x_n\} \subset V$ and let $\{\e_i\} \subset \mathbb R^n$ be the standard basis. Define an element $\x^* \in V\otimes \mathbb R^n$ as $\x^* = \sum_{i=1}^n \x_i \otimes \e_i$. Then for any $\P \in \mathcal P(V)$,
\begin{equation}
\sum_{\x\in X}(\P\x,\x) = (\P\otimes \I \, \x^*,\x^*).
\end{equation}
Let now introduce the sample mean
\begin{equation}
\widehat{\x} = \frac{1}{|X|} \sum_{\x\in X}\x,
\end{equation}
and $\bm{1} = [1,\dots,1]^\top \in \mathbb R^n$.

\begin{lemma}
\begin{equation}
	\S =
	\begin{pmatrix}
		{1-\frac{1}{n}}&{-\frac{1}{n}}&{\ldots}&{-\frac{1}{n}}\\
		{-\frac{1}{n}}&{1-\frac{1}{n}}&{\ldots}&{\vdots}\\	
		{\vdots}&{\vdots}&{\ddots}&{\vdots}\\	
		{-\frac{1}{n}}&{-\frac{1}{n}}&{\ldots}&{1-\frac{1}{n}}\\	
		\end{pmatrix} \in \mathbb{R}^{n \times n}
\end{equation}
is an orthogonal projector onto a subspace of codimension $1$.
\end{lemma}
\begin{proof}
The spectrum of $\S$ contains one $0$ eigenvalue and the rest $n-1$ eigenvalues are $1$-s.
\end{proof}

\begin{lemma}\label{lemma:distrChange}
Let $V$ be a vector space, $X\subset V$ be a finite mutually continuous set of vectors, and $\P \in\mathcal P(V)$ - a random operator, then there exists a set $Z\subset V$ of mutually continuous vectors such that $|Z| = |X| - 1$ and
\begin{equation}
\frac{1}{|X|}\sum_{\x\in X}(\P(\x- \widehat{\x}),\x- \widehat{\x}) = \frac{1}{|Z|}\sum_{\z\in Z}(\P\z,\z).
\end{equation}
\end{lemma}
\begin{proof}
Note that $\x^*-\widehat{\x} \otimes \bm{1} = (\I \otimes \S)\,\x^*$. Let $\{\f_i\}$ be another orthonormal basis of $\mathbb R^n$ such that $\ker \S = \langle \f_n\rangle$. Compute
\begin{equation}
(\I \otimes \S)\,\x^* = \sum_{i=1}^{n-1} \y_i \otimes \f_i,
\end{equation}
and denote $Y = \{\y_1,\ldots,\y_{n-1}\}$, then
\begin{equation}
\sum_{\x\in X}(\P(\x- \widehat{\x}),\x- \widehat{\x}) = ((\P\otimes \S)\, \x^*,\x^*) = (\P \otimes \I\, \y^*,\y^*) = \sum_{\y\in Y} (\P \y,\y),
\end{equation}
where $\y_i$ are now centered.

Take 
\begin{equation}
\z_i = \sqrt{\frac{|X|-1}{|X|}}\y_i,\quad i=1,\dots,n-1.
\end{equation}
The mutual continuity of $\z_i$ follows from the fact that the function $V^n \to V^{n-1}$ mapping $X$ to $Z$ is linear and surjective.
\end{proof}

\begin{lemma}\label{lemma:continuityExpVar}
Let $V$ and $U$ be vector spaces and $X\subset V\otimes U$ - a mutually continuous family of random vectors with $|X|> 1 + p/q+q/p$, then $f_\mathcal N(\widehat{\X},\P,\Q;X)$ as a function of $\P$ and $\Q$ extends to a continuous $\dt f_\mathcal N \colon \dt {\mathcal M}_\mathcal{N} \to \dt {\mathbb R}$ via $\dt f_\mathcal N (\infty) =+\infty$, in particular, $f_\mathcal N$ has a unique minimum over $\mathcal M_\mathcal{N}$.
\end{lemma}
\begin{proof}
Applying Lemma~\ref{lemma:distrChange}, we get a mutually continuous set $Z\subset V\otimes U$ such that $|Z| = |X| - 1$ and $f_\mathcal N(\widehat{\X},\P,\Q;X) = g_{\mathcal N}(\P,\Q;Z)$. Now the desired claim follows from Theorem~\ref{theorem:continuity}.
\end{proof}

\begin{theorem}\label{theorem:continuityExpVar}
Let $V$ and $U$ be vector spaces, $X\subset V\otimes U$ be a mutually continuous family of random vectors with $|X|> 1 + p/q+q/p$, and $S = V\otimes U \times \mathcal M_\mathcal{N}$, then $f_\mathcal{N}\colon S \to \mathbb R$ extends to a continuous function $\dt f_\mathcal{N} \colon \dt {S} \to \dt {\mathbb R}$ via $\dt f_\mathcal{N}(\infty) =+\infty$ and, in particular, there exists a unique minimum of $f_\mathcal{N}$ over $S$.
\end{theorem}
\begin{proof}
Note that for a fixed pair $(\P,\Q)$, the value of $\M$ minimizing $f_\mathcal{N}$ is the sample average, which does not depend on the values of $\P$ and $\Q$. Therefore, the result follows from Lemma~\ref{lemma:continuityExpVar} and Theorem~\ref{theorem:continuity}.
\end{proof}

\section{Proof of Theorem~\ref{tyl_mt}}
\label{RobustProofSec}
The proof in this section is quite similar to that given in \cite{soloveychik2014groups}, thus, we made it less verbose than the proof of the previous section. For more details please consult \cite{soloveychik2014groups}. Analogously to Definition \ref{def:Gflag}, we introduce the notion of a descending flag and note that the usage of flags in this section is different from that of Section \ref{GaussProofSec}.
\begin{definition}
Let $V$ be a real linear space, $X\subset V$ be a finite subset and $\mathcal F = \{V=V_0\supsetneq V_1\supsetneq \ldots V_s\supsetneq V_{s} \supseteq 0\}$ be a descending flag of length $s$ in $V$. Define
\begin{equation}
\Delta(\mathcal F, X)_{i\,j} = \dim V_{i} - \dim V_j -\frac{\dim V_1}{|X|}\left(|X\cap V_i| - |X\cap V_j|\right),
\end{equation}
where $0\leqslant i,\,j\leqslant s$. In addition, given a decreasing sequence
\begin{equation}
\r =\{r_1>\ldots>r_s\} \subset \mathbb{R}
\label{seq_r}
\end{equation}
of length $s$, define
\begin{equation}
S(\mathcal F, X, \r) = \sum_{i=1}^s r_i \Delta(\mathcal F,X)_{i-1\,i}.
\end{equation}
\end{definition}
It now follows immediately from the definition that
\begin{equation}
\Delta(\mathcal F, X)_{i\,j} + \Delta(\mathcal F, X)_{j\,k} = \Delta(\mathcal F, X)_{i\,k},\;\; i,j,k=0,\dots, s.
\label{saddpr}
\end{equation}

\begin{lemma}\label{lemma:PartitCoeff}
Let $X\subseteq V$ be a finite subset, $\mathcal F$ be a flag of length $s$ in $V$, $\r$ be a sequence as in (\ref{seq_r}), and $\Delta(\mathcal F, X)_{0\,i} < 0$ for all $i=1,\dots, s$. Then there exist a subflag $\mathcal F'\subseteq \mathcal F$ and a subsequence $\r'\subseteq \r$, both of length $t\leqslant s$, such that
\begin{equation}
S(\mathcal F,X,\r)\leqslant S(\mathcal F',X,\r'),
\end{equation}
\begin{equation}
\Delta(\mathcal F',X)_{i-1\,i} \leqslant 0,\;i=1,\dots,t.
\end{equation}
In particular, $S(\mathcal F,X,\r)< 0$.
\end{lemma}
\begin{proof}
The proof is by induction on $s$. For $s=1$,
\begin{equation}
S(\mathcal F,X,\r) = r_1 \Delta(\mathcal F,X)_{0\,1}<0.
\end{equation}
Let now $s>1$. If for all $i=1,\dots, s,\; \Delta(\mathcal F,X)_{i-1\,i}\leqslant 0$, then we are done since $\Delta(\mathcal F,X)_{0\,1} < 0$. Hence, we may assume that there is $i \leqslant s$ such that
\begin{equation}
\Delta(\mathcal F,X)_{j-1\,j}\leqslant 0,\; 1\leqslant j<i,\;\; \text{and}\;\; \Delta(\mathcal F,X)_{i-1\,i} > 0.
\end{equation}
Set $\mathcal F'$ to be $\mathcal F$ without $V_i$ and $\r'$ to be $\r$ without $r_i$, then
\begin{multline}
S(\mathcal F,X,\r) = \sum_{j\neq i-1,i}r_j \Delta(\mathcal F,X)_{j-1\,j} + r_{i-1}\Delta(\mathcal F,X)_{i-2\,i-1} + r_i \Delta(\mathcal F,X)_{i-1\,i} \\ \leqslant  \sum_{j\neq i-1,i}r_j \Delta(\mathcal F,X)_{j-1\,j} + r_{i-1}(\Delta(\mathcal F,X)_{i-2\,i-1} + \Delta(\mathcal F,X)_{i-1\,i}) = S(\mathcal F',X,\r'),
\end{multline}
where in the last equality we use (\ref{saddpr}). Since the length of $\mathcal F'$ is less than that of $\mathcal F$ and $\Delta(\mathcal F',X)_{0\,j}$ is either $\Delta(\mathcal F,X)_{0\,j-1}$ or $\Delta(\mathcal F,X)_{0\,j}$, thus, being negative, the result follows by induction.
\end{proof}

Let $V$ and $U$ be real vector spaces. For any $\V \in V\otimes U$, denote the subspace
\begin{equation}
\V U^* = \{\V \bm{\xi} \mid \bm{\xi} \in U^*\}\subseteq V,
\end{equation}
where $\V \bm\xi$ is the convolution along $U$.

\begin{lemma}\label{lemma:IntersectionCount}
Let $V$ and $U$ be vector spaces, and $X$ be a family of i.i.d. continuously distributed random vectors in $V\otimes U$, then
\begin{equation}
\dim \sum_{\X \in X}\X U^* = \min (|X| \dim U,\dim V),\;\;\text{a.s.}
\end{equation}
\end{lemma}
\begin{proof}
Choose bases in $V$ and $U$, then the elements of $X$ read as matrices and the space $\sum_{\X \in X} \X U^*$ is spanned by the columns of all $\X \in X$. Since the elements of $X$ are i.i.d. and continuously distributed, the matrix consisting of all columns of all $\X$-s is of maximal rank. Since it contains $|X|\dim U$ columns and $\dim V$ rows, the result follows.
\end{proof}

\begin{corollary}\label{corollary:IntersectionCount}
Let $V$ and $U$ be vector spaces, $K\subsetneq V$ be a proper subspace, and $X \subset V\otimes U$ be a family of i.i.d. continuously distributed vectors, then
\begin{equation}
|X \cap K\otimes U| \leqslant \frac{\dim K}{\dim U},\;\;\text{a.s.}
\end{equation}
\end{corollary}
\begin{proof}
Let $Y = X \cap K\otimes U$, then $E = \sum_{\Y\in Y} \Y U^*\subseteq K$ and Lemma~\ref{lemma:IntersectionCount} yields
\begin{equation}
|Y|\dim U=\dim E \leqslant \dim K.
\end{equation}
\end{proof}

Similarly to the Gaussian case, below we change the parametrization 
\begin{equation}
f_\mathcal{E}(\P,\Q;X) = \widetilde f_\mathcal{E}(\P^{-1},\Q^{-1};X),
\end{equation}
which does not affect the existence and uniqueness results. Let the dimensions of $V$ and $U$ be $p$ and $q$, correspondingly, and partition $f_\mathcal{E}(\P,\Q;X)$ as
\begin{equation}
f_\mathcal{E}(\P,\Q;X) = - \frac{1}{p}\ln|\P| - \frac{1}{q}\ln|\Q|
+ \frac{1}{n} \sum_{i=1}^n \ln\(\Tr{\P\X_i\Q\X_i^\top}\) = f_\P + f_\Q + f_X.
\label{fe_dev}
\end{equation}
Consider $f_\mathcal{E}(\P,\Q;X)$ over $\mathcal{M}_{\mathcal E}$ defined in (\ref{m_e_def}), which in our new notations means that
\begin{equation}
\Tr{\P^{-1}} = \Tr{\Q^{-1}} = 1.
\label{tr_inv_c}
\end{equation}
\begin{lemma}
\label{el_lem}
Let $V$ and $U$ be $p$ and $q$ dimensional real linear spaces, respectively, then if
\begin{equation}
|X| > \frac{\max(p,q)}{\min(p,q)},
\end{equation}
\begin{equation}
f_\mathcal{E} \to +\infty\;\; \text{as}\;\; \mathcal M_\mathcal{E} \ni (\P,\Q)\to \partial \mathcal M_\mathcal{E},\;\;\text{a.s.}
\end{equation}
\end{lemma}
\begin{proof}
Assume on the contrary, that there exists a sequence $\T=(\P,\Q) \subset \mathcal M_\mathcal{E}$ (we omit indices for brevity) tending to $\partial \mathcal M_\mathcal{E}$ and such that $f_\mathcal{E}(\T)$ is bounded. Note that due to (\ref{tr_inv_c}) at least a part of eigenvalues of $\P$ and $\Q$ tend to $+\infty$ and the others are bounded by positive constants from below.

Let $\P = \sum_{j=1}^p \lambda_j \y_j \y_j^\top$ and $\Q = \sum_{i=1}^q \mu_i \z_i \z_i^\top$ be the spectral decompositions of $\P$ and $\Q$. Passing to a subsequence, if needed, we may suppose everything to converge here. Below we do not mention explicitly the subsequence argument while it is assumed to be utilized if necessary. Denote the sets of eigenvalues of $\P$ and $\Q$ by $\Lambda$ and $\mathrm M$, respectively. Let $\Lambda = \sqcup_{i=1}^{u+1} \Lambda_i$ and $\rho$ be a sequence such that $\ln \lambda /\ln \rho \to r_i$ whenever $\lambda\in \Lambda_i$ and $r_1>\dots>r_u >r_{u+1} = 0$. In particular, $\ln \lambda \asymp r_i\ln \rho$ for $\lambda\in \Lambda_i$ (if $r_i=0$, by this expression we mean that $\ln \lambda = o(\ln \rho)$). Define $K_i$ to be the space generated by the limits of eigenvectors corresponding to the values in $\Lambda_i$, hence, $V = \oplus_{i=1}^{u+1} K_i$. Now set $V_i = (\oplus_{j=1}^i K_j)^\perp$ for all $i=0,\dots, u$. Define a flag $\mathcal F$ of length $u$ as
\begin{equation}
\mathcal F = \{V\otimes U = V_0\otimes U\supsetneq \ldots \supsetneq V_{u}\otimes U\}
\end{equation}
and $\r=\{r_1,\dots,r_{u}\}$.

In a similar way, let $\mathrm M = \sqcup_{j=1}^{v+1} \mathrm M_j$ and $\nu$ be such a sequence that $\ln \mu /\ln \nu \to t_j$ whenever $\mu \in \mathrm M_j$ and $t_1>\dots>t_v > t_{v+1} = 0$. In particular, $\ln \mu \asymp t_i \ln \nu$ for $\mu\in \mathrm M_i$. The space generated by the limits of eigenvectors corresponding to $\mathrm M_j$ will be denoted by $L_j$, hence, $U=\oplus_{j=1}^{v+1} L_j$. Now we set $U_i = (\oplus_{j=1}^iL_j)^\perp$ for all $i=0,\dots,v$. Define a flag $\mathcal G$ of length $v$ as
\begin{equation}
\mathcal G = \{V\otimes U = V\otimes U_0\supsetneq \ldots \supsetneq V\otimes U_{v}\}
\end{equation}
and $\t=\{t_1,\dots,t_v\}$.

Let $E_{ij} = V_{i-1}\otimes U_{j-1} \subseteq V\otimes U$, then for any $\lambda\in \Lambda_i,\; \mu\in \mathrm M_j$ and $\X\in V_{i-1}\otimes U_{j-1}$, the limit of
\begin{equation}
\frac{1}{\lambda} \P \X \frac{1}{\mu} \Q
\end{equation}
exists and will be denoted by $\R_{ij}(\X)$. By the definition, $\R_{ij}$ is a composition of the orthogonal projection onto $E_{ij}$ and a positive operator on the image. Let
\begin{equation}
X_{ij} = X\cap V_{i-1}\otimes U_{j-1} \setminus X\cap \( V_{i}\otimes U_{j-1} + V_{i-1} \otimes U_{j}\),
\end{equation}
then $X = \sqcup_{i=1}^{u+1}\sqcup_{j=1}^{v+1} X_{ij}$.

We now proceed to computing the leading asymptotic terms of the summands in (\ref{fe_dev}),
\begin{equation}
f_\P \asymp - \frac{1}{p}\sum_{i=1}^{u+1}\sum_{\lambda\in \Lambda_i} r_i \ln\rho = - \frac{1}{p}\sum_{i=1}^u r_i |\Lambda_i|\ln\rho =- \sum_{i=1}^u r_i \frac{\dim V_{i-1} - \dim V_{i}}{p}\ln\rho.
\end{equation}
Similarly,
\begin{equation}
f_\Q \asymp - \sum_{j=1}^v t_j \frac{\dim U_{j-1} - \dim U_{j}}{q}\ln\nu.
\end{equation}
Let $\X \in X_{ij}$, then for any $\lambda \in \Lambda_i$ and $\mu\in \mathrm M_j$, we have
\begin{multline}
\ln (\P \X \Q, \X) \asymp \ln \lambda + \ln \mu + \ln(\lambda^{-1} \P \X \mu^{-1} \Q, \X) \\ \asymp r_i\ln \rho + t_j \ln \nu + \ln (\R_{ij}(\X),\X) \asymp r_i\ln \rho + t_j \ln \nu.
\end{multline}
Taking this into account, we compute
\begin{equation}
f_X \asymp \frac{1}{|X|}\sum_{i=1}^{u+1}\sum_{j=1}^{r+1} \sum_{x\in X_{ij}}(r_i \ln\rho +t_i\ln \nu)
= \frac{1}{|X|}\sum_{i=1}^{u+1}\sum_{j=1}^{r+1} |X_{ij}|(r_i \ln\rho + t_i\ln \nu).
\end{equation}
We are interested in the asymptotic of the sum $f_\P+f_\Q+f_X$, whose leading term, when non-zero, can be written as
\begin{equation}
f_{\mathcal{E}} = f_\P+f_\Q+f_X \asymp A \ln \rho + B \ln \nu,
\end{equation}
where
\begin{align}
A &= - \sum_{i=1}^{u}r_i \frac{\dim V_{i-1} - \dim V_{i}}{p} + \frac{1}{|X|}\sum_{i=1}^{u+1}\sum_{j=1}^{r+1} |X_{ij}|r_i \nonumber \\
&= - \sum_{i=1}^{u}r_i\( \frac{\dim V_{i-1} - \dim V_{i}}{p} - \frac{\sum_{j=1}^{r+1}|X_{ij}|}{|X|}\) \nonumber\\
&= - \sum_{i=1}^{u}r_i\(\frac{\dim V_{i-1}\otimes U - \dim V_{i}\otimes U}{pq} - \frac{|X\cap V_{i-1}\otimes U| - |X \cap V_{i}\otimes U|}{|X|}\) \nonumber \\
&= - S(\mathcal F, X, \r).
\end{align}
Similar derivation yields $B = S(\mathcal G, X, \t)$. Thus,
\begin{equation}
f_\mathcal{E} \asymp - S(\mathcal F, X, \r)\ln \rho - S(\mathcal G,X,\t)\ln \nu,
\end{equation}
where the right-hand side is non-zero since at least one pair of eigenvalues tend to $+\infty$ due to the trace constraint (\ref{tr_inv_c}). In addition, this implies that $\ln \rho$ and $\ln \nu$ tend to $+\infty$, and it remains to show the coefficients $S(\cdot)$ are both negative, thus guaranteeing that $f_\mathcal{E} \to +\infty$. By Lemma~\ref{lemma:PartitCoeff}, it is enough to check that $\Delta(\mathcal F, X)_{0i}<0$ for $i=1,\dots,u$ and $\Delta(\mathcal G,X)_{0j}<0$ for $j=1\dots,v$. We have
\begin{multline}
\Delta(\mathcal F,X)_{0 i} = \frac{\dim V\otimes U - \dim V_i\otimes U}{pq} - \frac{|X| - |X\cap V_i\otimes U|}{|X|} \\ = \frac{|X\cap V_i\otimes U|}{|X|}-\frac{\dim V_i}{p},
\end{multline}
and we need to show
\begin{equation}
\frac{|X\cap V_i\otimes U|}{|X|}< \frac{\dim V_i}{p}.
\end{equation}
By Corollary~\ref{corollary:IntersectionCount}, we get
\begin{equation}
\frac{|X \cap V_i\otimes U|}{|X|} \leqslant \frac{\dim V_i}{q|X|} < \frac{\dim V_i}{\dim V},
\end{equation}
where the last inequality holds because $|X| > \max[p,q]/\min[p,q]$. After a similar calculation for $B$, we see that both $A$ and $B$ are negative and $f_{\mathcal{E}}(\T;X)\to+\infty$ as $\T \to \partial \mathcal M_\mathcal{E}$. This contradicts the boundedness assumption on $f_{\mathcal{E}}(\T;X)$ and completes the proof.
\end{proof}

\begin{proof}[Proof of Theorem~\ref{tyl_mt}]
Since the KP constraint and the target function are convex in the Riemannian metric introduced above, the only thing we need to show is that $f_{\mathcal{E}} \to + \infty$ when we approach $\partial\mathcal{M}_{\mathcal{E}}$. Lemma~\ref{el_lem} proves that this is a.s. true when
\begin{equation}
|X| > \frac{\max(p,q)}{\min(p,q)},
\end{equation}
and we are done.
\end{proof}

%

\bibliographystyle{plainnat}
\bibliography{ilya_bib}

\end{document}